\documentclass[11pt]{article}

\usepackage[margin=1in]{geometry}
\usepackage[T1]{fontenc}
\usepackage[utf8]{inputenc}
\usepackage{lmodern}
\usepackage{microtype}
\usepackage{amsmath,amssymb,amsthm,mathtools,bm}
\usepackage{graphicx}
\usepackage{booktabs}
\usepackage{array}
\usepackage{amscd}
\usepackage{tikz}
\usetikzlibrary{positioning,fit,backgrounds,arrows.meta}
\usepackage{tikz-cd}
\usepackage{pgfplots}
\pgfplotsset{compat=1.18}
\usepackage{caption}
\usepackage{subcaption}
\usepackage{enumitem}
\usepackage{cite}
\usepackage{hyperref}
\usepackage{xcolor}
\usepackage{float}
\usepackage{setspace}
\usepackage{multirow}
\usepackage{url}
\usepackage{placeins}

\hypersetup{
  colorlinks=true,
  linkcolor=blue!50!black,
  citecolor=blue!50!black,
  urlcolor=blue!60!black,
  pdftitle={Topological Ignorability, Weak Ignorability, and Structural Causal Effects Beyond Means},
  pdfauthor={Anonymous}
}

\newtheorem{theorem}{Theorem}
\newtheorem{proposition}[theorem]{Proposition}
\newtheorem{corollary}[theorem]{Corollary}
\newtheorem{lemma}[theorem]{Lemma}
\theoremstyle{definition}
\newtheorem{definition}[theorem]{Definition}
\newtheorem{remark}[theorem]{Remark}
\newtheorem{example}[theorem]{Example}

\newcommand{\R}{\mathbb{R}}
\newcommand{\E}{\mathbb{E}}
\newcommand{\Prob}{\mathbb{P}}
\newcommand{\Law}{\mathcal{L}}

\newcommand{\norm}[1]{\left\lVert #1 \right\rVert}
\DeclareMathOperator{\supp}{supp}
\DeclareMathOperator{\ATE}{ATE}
\DeclareMathOperator{\ESS}{ESS}
\DeclareMathOperator{\logit}{logit}

\title{\bfseries Topological Ignorability for Structural Causal Effects Beyond Means}
\author{%
	Usef Faghihi\thanks{Université du Québec à Trois-Rivières (UQTR).\\
		Email: usef.faghihi@uqtr.ca}
	\and
	Amir Saki\thanks{Université du Québec à Trois-Rivières (UQTR).\\
		Email: amir.saki@uqtr.ca}
}
\date{}

\begin{document}
\maketitle

\begin{abstract}
Many interventions change the structure of an outcome distribution rather than its mean: they may split a population into disconnected regimes, create loops or holes, generate new branches, or reorganize the geometry of an outcome cloud while leaving the average response nearly unchanged. In such settings, classical mean-based causal estimands such as the average treatment effect may fail to capture important structural changes.

To quantify such effects, we introduce topological-geometrical causal metrics based on topological and geometrical summaries of interventional outcome laws, including density-superlevel Betti summaries, Euler-based signatures, and
persistent-homology summaries. These metrics measure structural differences between treated and untreated outcome laws beyond averages.

We then study the identifying assumptions required for these metrics to admit a causal interpretation. We introduce topological ignorability, a topological analogue of classical conditional ignorability that requires invariance of the chosen structural feature rather than the full counterfactual distribution. When the chosen summary is injective, topological ignorability coincides with weak ignorability, the classical conditional ignorability assumption. For noninjective summaries, however, topological ignorability can identify the structural feature of interest without identifying the full interventional law.

To address this setting, we introduce a covariate-standardized topological-geometrical causal effect and develop practical estimation procedures. We validate the framework in two hidden-confounding benchmarks: a fully
synthetic exact benchmark and a real-covariate semi-synthetic benchmark using
Wisconsin breast-cancer covariates.  In both cases, weak ignorability fails and
balancing observed covariates nearly eliminates standardized mean differences,
yet the coordinate-mean average treatment effect remains biased.  In contrast,
selected finite density-superlevel Betti and Euler contrasts remain stable
across oracle, observational, and weighted analyses.

\end{abstract}

\textbf{Keywords:} causal inference; ignorability; topological data analysis; persistent homology; structural causal effect

\section{Introduction}

Classical causal inference is organized around the potential outcomes $Y(0)$ and $Y(1)$ and mean-based targets such as the average treatment effect (ATE) and the conditional average treatment effect (CATE) \cite{RosenbaumRubin1983,Holland1986,Pearl2009,ImbensRubin2015}. Those targets are exactly right when the scientific question is whether treatment changes the average response. But many contemporary applications ask a different question. A perturbation may create a new cell state, split one population into two phenotypic regimes, or stretch an outcome cloud into a filament while leaving its centroid almost unchanged. In labor and mental-health interventions, the relevant question may be whether the joint response cloud reorganizes into distinct subgroups rather than whether a single coordinate mean changes. In those settings, a mean effect can be zero while the structural causal effect is large and scientifically decisive.

This is precisely the regime in which topological summaries are useful. Persistent homology, persistence landscapes, Betti curves, Euler characteristic transforms, Euler integral transforms on constructible functions, and weighted Euler curve transforms are designed to capture connectivity, branching, multiscale structure, and location-sensitive shape \cite{Bubenik2015,Turner2014,Baryshnikov2011,Ghrist2018,Curry2022,Jiang2020}. In this paper, we use \(\Psi\) to denote a topological or geometrical summary map that assigns to an outcome law a structural signature derived from its density. To avoid ambiguity, we work with a concrete geometric representation: for each law $\mu$, we consider its density $f_\mu$ (with respect to a fixed reference measure). We assume throughout that all laws considered admit densities with respect to this reference measure.  In continuous models, the ordinary support is often the entire ambient space (e.g., for Gaussian laws on $\mathbb{R}^d$), and therefore carries no useful topological information. Instead, topological structure is extracted from the geometry of the density $f_\mu$, for example through its superlevel-set filtration $\{x : f_\mu(x) \ge \lambda\}$ or related filtration-based constructions. Thus, throughout the paper, $\Psi(\mu)$ should be interpreted as a topological summary derived from the geometry of the density $f_\mu$.  The causal challenge is not defining such summaries, but determining the weakest identifying assumptions under which they admit a causal interpretation, and characterizing which topological-geometrical estimands remain identifiable when $\Psi$ is noninjective. This idea is illustrated in Figure~\ref{fig:filtration}, which shows the density-induced filtration given by nested superlevel sets and the resulting evolution of connectivity across levels. Therefore, $\Psi(\mu)$ depends on the entire filtration $\{x : f_\mu(x)\ge \lambda\}_\lambda$, rather than on any single level.
\begin{figure}[t]
	\centering
\begin{tikzpicture}[
	scale=0.8,
	axis/.style={->, line width=0.6pt, black!80},
	outer/.style={fill=blue!12, draw=blue!55!black, line width=0.65pt, opacity=0.75},
	middle/.style={fill=cyan!25, draw=cyan!50!black, line width=0.75pt, opacity=0.78},
	inner/.style={fill=violet!55, draw=violet!60!black, line width=0.85pt, opacity=0.88},
	labelbox/.style={
		draw=#1,
		rounded corners=2.5pt,
		fill=white,
		line width=0.55pt,
		inner xsep=5pt,
		inner ysep=2pt,
		font=\large,
		text=#1
	},
	leader/.style={line width=0.55pt, -{Latex[length=2mm]}, #1}
	]
	
	
	\draw[axis] (-4.2,0) -- (4.2,0) node[right] {$x$};
	\draw[axis] (0,-2.05) -- (0,2.05) node[above] {$y$};
	
	\fill[outer] (-1.5,0) ellipse (2.2 and 1.3);
	\fill[outer] ( 1.5,0) ellipse (2.2 and 1.3);
	
	\fill[middle] (-1.5,0) ellipse (1.4 and 0.9);
	\fill[middle] ( 1.5,0) ellipse (1.4 and 0.9);
	
	\fill[inner] (-1.5,0) ellipse (0.7 and 0.5);
	\fill[inner] ( 1.5,0) ellipse (0.7 and 0.5);
	
	\draw[blue!60!black, dashed, line width=0.75pt] 
	(-1.5,0) ellipse (2.2 and 1.3);
	\draw[blue!60!black, dashed, line width=0.75pt] 
	(1.5,0) ellipse (2.2 and 1.3);
	
	\draw[cyan!45!black, line width=0.8pt] 
	(-1.5,0) ellipse (1.4 and 0.9);
	\draw[cyan!45!black, line width=0.8pt] 
	(1.5,0) ellipse (1.4 and 0.9);
	
	\draw[violet!60!black, line width=0.95pt] 
	(-1.5,0) ellipse (0.7 and 0.5);
	\draw[violet!60!black, line width=0.95pt] 
	(1.5,0) ellipse (0.7 and 0.5);
	
	\node[labelbox=blue!60!black] (L1) at (-3.75,1.95) {$\lambda_1$};
	\node[labelbox=cyan!45!black] (L2) at (-3.75,1.25) {$\lambda_2$};
	\node[labelbox=violet!60!black] (L3) at (-3.75,0.55) {$\lambda_3$};
	
	\draw[leader=blue!60!black] 
	(L1.east) -- (-2.05,1.3);
	
	\draw[leader=cyan!45!black] 
	(L2.east) -- (-2 ,0.85);
	
	\draw[leader=violet!60!black] 
	(L3.east) -- (-2.05,0.25);

\end{tikzpicture}
	
\caption{
	Filtration induced by the density $f_\mu$. Nested shaded regions depict superlevel sets $\{x : f_\mu(x)\ge \lambda\}$ for decreasing thresholds. Darker regions represent higher density. At high levels, two disconnected components are present; as the threshold decreases, these regions expand and eventually merge into a single connected set, revealing the multiscale topological structure captured by summaries such as persistent homology, Betti curves, and Euler characteristic.
}
	\label{fig:filtration}
\end{figure}

The key distinction should appear immediately. Weak ignorability, also known as the classical conditional ignorability assumption in causal inference, controls the \emph{entire} arm-specific counterfactual law within each covariate stratum:
\[
\Law(Y(t)\mid T=t,Z=z)=\Law(Y(t)\mid Z=z).
\]
Topological ignorability relative to a chosen transform or functional $\Psi$ asks only that the \emph{chosen structural target} agree:
\[
\Psi\!\bigl(\Law(Y(t)\mid T=t,Z=z)\bigr)=\Psi\!\bigl(\Law(Y(t)\mid Z=z)\bigr).
\]
Intuitively, weak ignorability says ``the whole law is unchanged by treatment selection once we condition on $Z$''; topological ignorability says only ``the structural feature we care about is unchanged.'' If $\Psi$ is injective, these are the same statement written in different languages. If $\Psi$ is coarse and noninjective, topological ignorability can be strictly weaker than weak ignorability, in the sense that it may hold even when the full counterfactual law is not identified. Hidden confounding may alter densities (via changes in mixture weights or component structure) without changing whether the high-density regions of the treated law have one connected component or two. When the scientific target is exactly that bifurcation, the weaker target-dependent notion is the relevant one.

The noninjective case is not a corner case. In many applications $\Psi$ is a persistent-homology-type summary, a persistence landscape, a Betti curve, a persistence image, or a connectivity functional applied to the density via its superlevel-set filtration.  Those transforms are typically not injective as maps from outcome laws to their topological signatures.  Different laws can yield identical topological summaries, such as the same persistence diagram, the same persistence landscape, or the same number of connected components in the high-density regions of the distribution. In that setting, topological ignorability does not imply weak ignorability, and one should not expect identification of the full interventional laws. The correct response is not to abandon topological targets, but to focus on estimands that remain identifiable under topological ignorability.  Formally, we view $\Psi$ as a functional of the density $f_\mu$,
\[
\Psi(\mu) = \Phi(f_\mu),
\]
where $\Phi$ extracts topological information from the geometry of $f_\mu$, for example via its superlevel-set filtration or related constructions. This perspective makes explicit that causal identification operates on geometric features of the density rather than on the full law $\mu$, and that the choice of topological summary determines both the estimand and the identifying assumptions required to identify and estimate it from observed data. In particular, different probability laws may induce densities whose geometric structures are topologically indistinguishable under the chosen summary $\Phi$, and hence identical topological summaries, reflecting the inherent noninjectivity of $\Psi$.

Figure~\ref{fig:sam} gives a geometric picture in terms of high-density regions of the outcome law. The control arm exhibits a single connected region. Intervention~1 deforms this region into an elongated or filament-like structure, while Intervention~2 produces a differently shaped but still connected cloud. In both cases, the distribution remains connected, but its geometric organization changes substantially. After recentering, the centroids may coincide, so a mean-based estimand can be negligible even though the structural consequences are very different. This illustrates that causal effects may manifest as geometric deformations of the distribution, even when classical mean-based summaries are unchanged, and that the relevant signal lies in the geometry of the distribution rather than in its mean or in its full measure-theoretic support.
\begin{figure}[htbp]
  \centering
  \includegraphics[ scale =0.5]{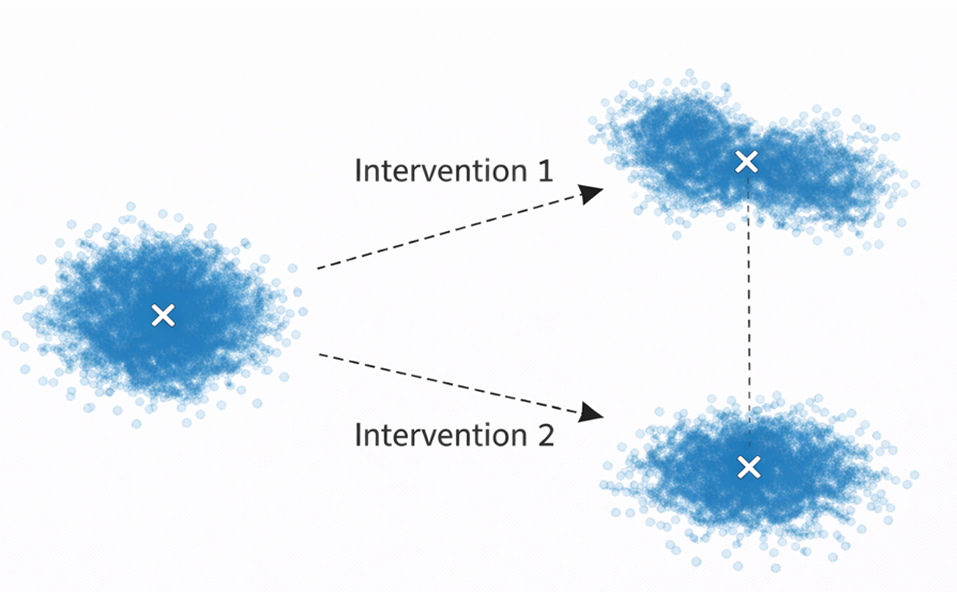}
\caption{Motivating structural effects beyond means. The control law (left) is connected. Intervention~1 deforms the distribution into an elongated or filament-like structure, while Intervention~2 produces a differently shaped but still connected cloud. Mean-based targets can be near zero even when the structural contrast is substantial.}
  \label{fig:sam}
\end{figure}

The paper makes five contributions. First, it formalizes a law-level topological-geometrical contrast, where \(\Psi(P_t)\) takes values in a normed space \(E\), and \(\norm {\;\cdot\;}_E\) denotes the corresponding norm:
\[
\Theta_{\Psi}=\norm{\Psi(P_1)-\Psi(P_0)}_E,\qquad P_t=\Law(Y(t)),
\]
together with a conditional contrast and a \emph{covariate-standardized} topological effect
\[
\tau_{\Psi}=\int \norm{\Psi(P_{1,z})-\Psi(P_{0,z})}_E\,dP_Z(z).
\]
Second, it shows that on injective model classes, topological ignorability is equivalent to weak ignorability; topology then changes the \emph{estimand}, not the identifying assumptions. Third, it shows why the noninjective regime is different: topological ignorability can identify 
 $\tau_{\Psi}$ even when it does not identify $\Theta_{\Psi}$ or the full counterfactual law. Fourth, it gives a fully explicit hidden-confounding SEM example with
 closed-form densities, posterior weights, and a second SEM with the same
 observational law but a different ATE.
 Fifth, it restores the estimation layer needed for applications with noninjective transforms: conditional plug-in consistency, a practical score-bin estimator, finite-signature approximation, overlap diagnostics, and a simple hidden-bias stress test.
 
\paragraph{Related work and positioning.}
This paper is related to three strands of work.  First, several causal-inference
literatures replace the ordinary average treatment effect by richer targets,
including distributional treatment effects and quantile treatment effects
\cite{Firpo2007,Chernozhukov2013}, kernel or metric comparisons of distributions
\cite{GrettonEtAl2012,SzekelyRizzo2013}, and causal inference with functional
or random-object outcomes \cite{Ecker2024,Santambrogio2015}.  These approaches
share the idea that the causal estimand should match the scientific question
rather than defaulting to \(\E[Y(1)-Y(0)]\).  Our target differs because the
summaries considered here are topological or geometric features of
density-induced filtrations.

Second, topology has also appeared in causal inference in work on topological
perspectives on causal structure, causal homotopy, topology-based causal
analysis of dynamical systems, topology-aware representation balancing, and
persistent-homology-based causal effects
\cite{ibeling_icard_2021,mahadevan_2021,bando_kaji_yaguchi_2022,farzam_2025,kim_lee_tce_2026}.
In most of these settings, topology is used to study causal structure, support
causal discovery, stabilize representation learning, or propose specific
topological statistics.  By contrast, our focus is a general law-level
identification framework: the estimand is itself topological, and the main
question is when a chosen structural feature of the interventional outcome law
is identified under weak or topological ignorability.

Third, statistical topological data analysis provides many of the summaries
used here, including persistence landscapes, Betti curves, persistence images,
Euler characteristic transforms, and weighted Euler curve transforms
\cite{Bubenik2015,Adams2017,Baryshnikov2011,Turner2014,Ghrist2018,Curry2022,Jiang2020}.
These tools provide candidate transforms \(\Psi\) and norms \(\norm{\cdot}_E\),
but they do not by themselves give a causal interpretation.  The contribution
of the present paper is to connect such summaries to causal identification
through consistency, positivity, weak ignorability, and the target-dependent
topological ignorability condition introduced here.

The central point is not that topology provides another descriptive summary of
data.  Rather, the choice of a topological transform \(\Psi\) changes the causal
target and therefore changes the relevant identification question.  The paper
asks when a structural feature of the interventional law is identified even
when the full interventional law is not.  This places the framework within the
standard causal-inference logic of targets, assumptions, identification,
estimation, and sensitivity analysis.

The paper is intentionally focused. Because the page budget is tight, the
emphasis is on conditional topological ignorability, noninjective transforms,
the worked hidden-confounding example, and the estimable target that remains
available when weak ignorability is too strong for the scientific question.

\section{Preliminaries: topology, persistence, and Euler signatures}
\label{sec:preliminaries}
This section gives a compact summary of the topological objects used later.
The paper treats an outcome law through the geometry of its density, rather
than only through coordinates, moments, or ordinary support.  In continuous
models, ordinary supports may be topologically uninformative, so the relevant
spaces are density-induced superlevel sets
\[
\{y:f_\mu(y)\ge c\}.
\]
Topological summaries then record how connected components, holes, Euler
characteristics, or directional Euler scans evolve across density levels.

\paragraph{Superlevel sets, components, and Betti numbers.}
Although the support \(\supp(P)\) of a probability law is a standard
topological object, it is often too coarse for the continuous models considered
in this paper.  For example, Gaussian and Gaussian-mixture laws on
\(\mathbb R^d\) may have full support, even when their densities have very
different modal or high-density geometry.  We therefore work primarily with
density superlevel sets.  For a law \(\mu\) with density \(f_\mu\), and for a
density level \(c\), define
\[
A_c(\mu):=\{y\in\mathcal Y:f_\mu(y)\ge c\}.
\]
As \(c\) varies, these sets form a filtration of high-density regions.  The
zeroth Betti number \(\beta_0(A_c(\mu))\) counts the number of connected
components of the superlevel set.  Thus \(\beta_0=1\) for one connected
high-density region and \(\beta_0=2\) for two separated high-density regions.
The first Betti number \(\beta_1(A_c(\mu))\) counts one-dimensional holes or
loops, and higher Betti numbers count higher-dimensional voids.

The Euler characteristic is a coarser alternating summary of these Betti
numbers:
\[
\chi(A_c(\mu))
=
\beta_0(A_c(\mu))-\beta_1(A_c(\mu))+\beta_2(A_c(\mu))-\cdots .
\]
In the two-dimensional examples in this paper, a connected high-density bar has
one component and no hole, so \(\chi=1\); two separated high-density blobs have
two components and no holes, so \(\chi=2\).  A finite superlevel-set effect of
\(+1\) means that, at the selected density levels, treatment changes the
high-density topology from one connected regime to two regimes, independently
of the coordinate mean.

\paragraph{Filtrations and persistent homology.}    Persistent homology is a multiscale version of these counts. Early persistence-like ideas appeared in Frosini’s Size Theory and size functions in the early 1990s, particularly for connected components \(H_0\) \cite{Frosini1990,Frosini1992,Verri1993}. Modern persistent homology was later formalized in computational topology and topological data analysis \cite{ZomorodianCarlsson2005,EdelsbrunnerHarer2010,Carlsson2009}. It has also been extended to statistical summaries such as persistence landscapes \cite{Bubenik2015} and linked to Euler-based transforms \cite{Ghrist2018}.

 Starting from a point cloud, one builds a nested family of spaces
\[
K_{r_1}\subseteq K_{r_2}\subseteq\cdots,
\]
where the scale $r$ may be the radius of balls around data points, a threshold in a density estimator, or a sublevel/superlevel threshold of a function. Homology is computed at every scale. A feature is ``born'' when it appears and ``dies'' when it merges with another feature or is filled in. The collection of birth--death pairs is a persistence diagram or barcode. Long bars are interpreted as stable geometric signal; very short bars are often treated as sampling noise. The $H_0$ part of persistent homology tracks connected components. In a cloud with two well-separated blobs, the two components remain distinct until a relatively large scale, producing a large $H_0$ death time. This is why the experiments report an $H_0$ largest-gap or death-scale diagnostic: it is a simple numerical proxy for the separation between components.

\paragraph{Finite summaries of persistence.}
Persistence diagrams are rich objects, but statistical analysis often uses finite or functional summaries. Persistence landscapes, Betti curves, and persistence images convert a diagram into an object that can be averaged, compared, or inserted into a regression-style estimator \cite{Bubenik2015}. In this paper such summaries are represented abstractly by a transform $\Psi(P)$. The transform may be noninjective: many different laws can have the same number of components, the same Betti curve, or even the same finite persistence summary. This loss of information is not a defect when the scientific target is deliberately coarse. It is exactly what allows a topological estimand to ignore density changes that are irrelevant to the structural question, while a mean estimand may remain sensitive to those changes.

\paragraph{Finite density-superlevel signatures.}
In applications, the density \(f_\mu\) is not observed directly.  It is
estimated from a finite sample, usually on a finite grid and with a smoothing
parameter.  The superlevel filtration is therefore replaced by a finite
collection of estimated superlevel sets
\[
\{\widehat f \ge c\},
\qquad c\in\Lambda,
\]
where \(\Lambda\) is a finite set of density levels.  The resulting Betti
curves, Euler curves, persistence summaries, and finite ECT vectors should be
understood as numerical approximations to the corresponding population
signatures.  This is why the empirical section reports finite-grid summaries
together with robustness diagnostics over grid size, smoothing, level choices,
and small-component thresholds.

\paragraph{Euler characteristic transform.}
The Euler characteristic transform (ECT) is a location-sensitive signature of shape \cite{Baryshnikov2011,Turner2014,Ghrist2018,Curry2022}. For each direction $v$ on the unit sphere and each threshold $a$, one looks at the part of a shape whose projection onto $v$ is at most $a$:
\[
S_{v,a}=\{x\in S: \langle x,v\rangle\le a\}.
\]
The ECT records
\[
\operatorname{ECT}_S(v,a)=\chi(S_{v,a}),
\]
that is, the Euler characteristic of these directional slices as the threshold moves. Intuitively, one scans the shape from many directions and records when components and holes appear or disappear. Because the scan is directional, the ECT contains more spatial information than the single number $\chi(S)$. In computations, only finitely many directions and thresholds are used, so the ECT becomes a finite vector signature. Comparing two ECT vectors gives a practical metric for shape change.

\paragraph{How these concepts enter the causal argument.}
The paper does not claim that topology identifies every causal quantity. Instead, it asks which causal quantity is being targeted. If $\Psi$ is injective on the model class, topological ignorability is equivalent to weak ignorability because equality of $\Psi$ forces equality of the full conditional law. If $\Psi$ is noninjective, equality of $\Psi$ can hold even when hidden confounding changes unobserved mixture weights, coordinate means, or densities. The estimable target is then the covariate-standardized contrast in the chosen signature, not the full counterfactual law. The finite superlevel-set experiments illustrate this point directly: hidden
\(U\) biases the coordinate-mean ATE, but the selected Betti and Euler
summaries remain stable because \(U\) changes sampling density within the
outcome regions without changing the reported finite superlevel-set topology.

\section{Setup, estimands, and identification}

Work on a complete probability space $(\Omega,\mathcal A,\Prob)$. Let $T\in\{0,1\}$ be treatment, let $Z$ be observed covariates in a standard Borel space $\mathcal Z$, and let $Y(0),Y(1)$ be potential outcomes taking values in $\mathcal Y\subseteq\R^d$. Consistency holds:
\[
Y=Y(T)\qquad\text{a.s.}
\]
For $t,s\in\{0,1\}$ and covariate value $z$, write
\[
P_{t,s,z}:=\Law(Y(t)\mid T=s,Z=z),\qquad P_{t,z}:=\Law(Y(t)\mid Z=z),
\]
and let $P_t=\Law(Y(t))$ denote the marginal interventional law. When a dominated model is useful, we assume all relevant conditional laws admit densities with respect to a common $\sigma$-finite measure $\lambda$ on $\mathcal Y$ and we write $f_{t,s,z}$ and $f_{t,z}$ for the corresponding densities.

\begin{definition}[Topological causal contrasts]
Let $\Phi$ be a measurable map from a class of densities on $\mathcal Y$ into a Banach space $(E,\norm{\cdot}_E)$. We define
	\[
	\Psi(\mu) := \Phi(f_\mu),
	\]
	where $f_\mu$ denotes the density of $\mu$.
 The \emph{marginal topological causal metric} is
\[
\Theta_{\Psi}:=\norm{\Psi(P_1)-\Psi(P_0)}_E.
\]
Its conditional version is
\[
\theta_{\Psi}(z):=\norm{\Psi(P_{1,z})-\Psi(P_{0,z})}_E.
\]
The \emph{covariate-standardized topological causal effect} is
\[
\tau_{\Psi}:=\int \theta_{\Psi}(z)\,dP_Z(z)=\E[\theta_{\Psi}(Z)].
\]
When $E=\R$ and $\Phi$ is taken to be a connectivity-based functional (e.g., the Betti curve $\lambda \mapsto \beta_0(\{x : f_\mu(x)\ge \lambda\})$ or a persistence summary), the contrast records changes in the connectivity structure of the density across levels of the induced filtration.
\end{definition}
This formulation makes explicit that causal contrasts are defined in terms of topological features of the density $f_\mu$ rather than the full probability law itself. Intuitively, $\Theta_\Psi$ compares the global topological structure of the treated and untreated interventional laws, while $\theta_\Psi(z)$ compares this structure within a fixed covariate stratum. The quantity $\tau_\Psi$ aggregates these within-stratum differences over the covariate distribution.

\begin{definition}[Weak, marginal topological, and conditional topological ignorability]
	Weak ignorability means
	\[
	Y(t)\perp T\mid Z,\qquad t\in\{0,1\}.
	\]
Marginal topological ignorability relative to $\Psi$ is the condition
	\[
	\Psi\!\bigl(\Law(Y(t)\mid T=t)\bigr)=\Psi(P_t),
	\qquad t\in\{0,1\}.
	\]
	This condition targets identification of the marginal contrast $\Theta_\Psi$.
	
	Conditional topological ignorability relative to $\Psi$ means
	\[
	\Psi(P_{t,t,z})=\Psi(P_{t,z})
	\qquad\text{for }t\in\{0,1\}\text{ and }P_Z\text{-a.e. }z.
	\]
	This condition is tied to the covariate-standardized effect $\tau_\Psi$.
\end{definition}

Thus weak ignorability is a law-level condition, whereas conditional
topological ignorability is a target-level condition.  The latter should not be
read as a weaker universal assumption; it is weaker only relative to the chosen
noninjective summary \(\Psi\).  Changing \(\Psi\) changes both the estimand and
the identification burden.

\begin{proposition}[Identification under weak ignorability]
\label{prop:gformula}
Assume positivity,
\[
0<e(z):=\Prob(T=1\mid Z=z)<1\qquad\text{a.e.},
\]
and weak ignorability. Then for every measurable $B\subseteq\mathcal Y$,
\[
P_t(B)=\int \Prob(Y\in B\mid T=t,Z=z)\,dP_Z(z).
\]
Consequently, $P_t$, $\Theta_{\Psi}$, $\theta_{\Psi}(z)$, and $\tau_{\Psi}$ are identified whenever $\Psi$ is measurable on the identified law.
\end{proposition}

\begin{proof}
Weak ignorability implies $P_{t,z}=P_{t,t,z}$ almost surely; consistency identifies $P_{t,t,z}$ with $\Law(Y\mid T=t,Z=z)$. Integrating over $Z$ yields the usual g-formula, and applying the deterministic map $\Psi$ identifies the displayed effect functionals.
\end{proof}

\begin{remark}[Mixtures must be handled before the transform]
	\label{rem:noncommute}
	For nonlinear topological transforms, $\Psi$ generally does not commute with mixing:
	\[
	\Psi\!\left(\int P_{t,z}\,dP_Z(z)\right)\neq\int \Psi(P_{t,z})\,dP_Z(z).
	\]
	Thus the topology of a mixture cannot, in general, be recovered from the topologies of its components.
\end{remark}

\begin{theorem}[Injective transforms collapse topological ignorability to weak ignorability]
\label{thm:exact}
Assume positivity and suppose $\Psi$ is injective on the model class. Then conditional topological ignorability is equivalent to weak ignorability.

\end{theorem}

\begin{proof}
Weak ignorability gives $P_{t,t,z}=P_{t,z}$ almost surely, hence applying $\Psi$ yields topological ignorability.

Conversely, suppose conditional topological ignorability holds and $\Psi$ is injective. Then $P_{t,t,z}=P_{t,z}$ almost surely for $t=0,1$. In density notation,
\[
f_{t,z}=e(z)f_{t,1,z}+(1-e(z))f_{t,0,z}.
\]
For $t=1$, the equality $f_{1,1,z}=f_{1,z}$ implies
\[
(1-e(z))f_{1,0,z}=(1-e(z))f_{1,z},
\]
so positivity yields $f_{1,0,z}=f_{1,z}$. Thus $f_{1,0,z}=f_{1,1,z}=f_{1,z}$. The argument for $t=0$ is symmetric: $f_{0,0,z}=f_{0,z}$ implies $f_{0,1,z}=f_{0,z}$. Therefore the conditional law of $Y(t)$ does not depend on treatment once $Z$ is fixed, which is exactly weak ignorability.
\end{proof}

\begin{proposition}[Identification under conditional topological ignorability for noninjective transforms]
\label{prop:tau-ident}
Assume positivity and conditional topological ignorability relative to $\Psi$. Then for $P_Z$-a.e.\ $z$,
\[
\theta_{\Psi}(z)=
\norm{\Psi\!\bigl(\Law(Y\mid T=1,Z=z)\bigr)-\Psi\!\bigl(\Law(Y\mid T=0,Z=z)\bigr)}_E.
\]
Consequently the standardized effect
\[
\tau_{\Psi}=\int \theta_{\Psi}(z)\,dP_Z(z)
\]
is identified by the observational law of $(Y,T,Z)$.
\end{proposition}

\begin{proof}
Conditional topological ignorability gives
\[
\Psi(P_{t,z})=\Psi(P_{t,t,z})\qquad\text{a.s. for }t=0,1.
\]
By consistency,
\[
P_{t,t,z}=\Law(Y\mid T=t,Z=z).
\]
Substituting these identities into the definition of $\theta_{\Psi}(z)$ proves the displayed representation, and integrating over $P_Z$ identifies $\tau_{\Psi}$.
\end{proof}
This result shows that in the noninjective regime, causal identification operates at the level of topological summaries rather than at the level of the full interventional laws.
\begin{remark}[Why $\Theta_{\Psi}$ may fail to be identified under conditional  topological ignorability alone]
\label{rem:theta-not-id}
Proposition~\ref{prop:tau-ident} identifies the transformed \emph{conditional} laws, not the marginal interventional law itself. Unless weak ignorability holds or $\Psi$ commutes with mixing, conditional topological ignorability does \emph{not} identify
\[
\Theta_{\Psi}=\norm{\Psi(P_1)-\Psi(P_0)}_E.
\]
This is the fundamental reason to target $\tau_{\Psi}$ rather than $\Theta_{\Psi}$ in the noninjective regime.
\end{remark}

\begin{table}[htbp]
	\centering
	\small
	\renewcommand{\arraystretch}{1.15}
	\caption{Identification of topological causal targets under different assumptions. Injectivity of $\Psi$ determines whether topological equality implies equality of laws.}
	\label{tab:identification-summary}
	\begin{tabular}{>{\raggedright\arraybackslash}p{3.2cm} >{\raggedright\arraybackslash}p{3.1cm} >{\raggedright\arraybackslash}p{6.2cm}}
		\toprule
		Assumptions & Properties of $\Psi$ & Identified topological target \\
		\midrule
		
		Weak ignorability $+$ positivity 
		& Arbitrary measurable $\Psi$ 
		& The interventional laws $P_t$, the marginal contrast $\Theta_{\Psi}$, the conditional contrast $\theta_{\Psi}(z)$, and the standardized effect $\tau_{\Psi}$ \\
		
		\addlinespace[0.2em]
		
		Marginal topological ignorability $+$ positivity 
		& Injective $\Psi$ 
		& The marginal interventional laws $P_t$ and the marginal contrast $\Theta_{\Psi}$ \\
		
		\addlinespace[0.2em]
		
		Marginal topological ignorability $+$ positivity 
		& Noninjective $\Psi$ 
		& The marginal contrast $\Theta_{\Psi}$ only \\
		
		\addlinespace[0.2em]
		
		Conditional topological ignorability $+$ positivity 
		& Injective $\Psi$ 
		& The interventional laws $P_t$, the marginal contrast $\Theta_{\Psi}$, the conditional contrast $\theta_{\Psi}(z)$, and the standardized effect $\tau_{\Psi}$ \\
		
		\addlinespace[0.2em]
		
		Conditional topological ignorability $+$ positivity 
		& Noninjective, nonlinear $\Psi$ 
		& The conditional contrast $\theta_{\Psi}(z)$ and the standardized effect $\tau_{\Psi}$; the marginal contrast $\Theta_{\Psi}$ is not identified in general \\
		
		\bottomrule
	\end{tabular}
\end{table}

\begin{example}[Persistent homology and nonidentification of $\Theta_\Psi$]
	Let $\Psi_{\mathrm{PH}}(\mu)=\Phi_{\mathrm{PH}}(f_\mu)$ be a persistent homology summary derived from the superlevel-set filtration of the density $f_\mu$. Such summaries are noninjective: distinct probability laws may induce identical persistence diagrams.
	
	Consider a covariate $Z\in\{0,1\}$ with $P(Z=0)=P(Z=1)=1/2$. Suppose that for each $z$, the conditional laws under treatment and control have densities whose high-density regions consist of two well-separated modes, so that $
	\Psi_{\mathrm{PH}}(P_{t,z})$ indicates two persistent components
for all $t,z$.
Thus the conditional contrasts $\theta_{\Psi_{\mathrm{PH}}}(z)$ are identified from the observational law.
	
	Now consider two different marginal interventional laws obtained by mixing these conditional distributions:
	
	\begin{itemize}
		\item In Model A, the two modes corresponding to $z=0$ and $z=1$ are aligned in the same locations, so that the mixture still exhibits two well-separated modes.
		\item In Model B, the modes are shifted so that the high-density regions overlap across $z$, producing a single connected high-density region in the mixture.
	\end{itemize}
	
	In both models, the conditional summaries $\Psi_{\mathrm{PH}}(P_{t,z})$ are identical and hence observationally indistinguishable under conditional topological ignorability. However, the marginal summaries differ:
	\[
	\Psi_{\mathrm{PH}}(P_t^{\mathrm{A}})\neq \Psi_{\mathrm{PH}}(P_t^{\mathrm{B}}),
	\]
	since the topology of the mixture changes.
	Hence, the marginal contrast $\Theta_{\Psi_{\mathrm{PH}}}$ is not identified, whereas the conditional contrasts $\theta_{\Psi_{\mathrm{PH}}}(z)$ and their average $\tau_{\Psi_{\mathrm{PH}}}$ remain identifiable.
\end{example}
This example illustrates that topological summaries of mixtures depend on the geometric alignment of conditional distributions, which is not identified from observational data.

\section{SEM-Level Sufficient Conditions for Conditional Topological Ignorability}
\label{sec:sem-sufficient}

Conditional topological ignorability was defined above as the equality
\[
\Psi(P_{t,t,z})=\Psi(P_{t,z}),
\qquad t\in\{0,1\},
\]
inside almost every covariate stratum.  This section gives structural conditions under which this equality follows from an SEM.  The key mechanism is simple: conditioning on factual treatment may reweight latent regimes, so weak ignorability can fail, but the reweighted density may remain in the same topological chamber as the interventional density.  In that case the law changes, while the topological signature chosen by \(\Psi\) does not. 

The purpose of the SEM-based results is not to provide directly testable criteria for conditional topological ignorability, but rather to explain the structural mechanisms under which such topological stability can naturally arise under hidden confounding.  As in classical causal inference, the interventional laws \(P_{t,z}\) are generally unobserved, so conditional topological ignorability cannot typically be verified directly from observational data alone.  Instead, the SEM framework provides a mechanistic interpretation: latent confounding induces a reweighting of latent regimes, which deforms the observational density relative to the interventional density.  The results below characterize situations in which this deformation preserves selected topological features---such as connected components, loops, persistence barcodes, or Euler summaries—even though the full probability law changes and weak ignorability may fail. In this sense, they should be interpreted as structural explanations for why certain topological summaries may remain identifiable despite confounded probabilistic deformation.

\subsection{Latent reweighting induced by an SEM}
\label{subsec:sem-latent-reweighting}

Assume an SEM of the form
\[
T=h_T(Z,U,\varepsilon_T),
\qquad
Y(t)=g_t(Z,U,\varepsilon_t),
\qquad t\in\{0,1\},
\]
where \(U\) is a latent variable.  We assume that, conditional on \((Z,U)\), the outcome noise \(\varepsilon_t\) is not affected by the treatment-assignment noise.  Equivalently, the conditional law of \(Y(t)\) given \((Z,U,T)\) is the same as its conditional law given \((Z,U)\).  This is the usual exogenous-noise condition in an SEM; the treatment may depend on \(U\), but it does not change the latent-regime kernel for the potential outcome.

For fixed \((t,z,u)\), define
\[
\kappa_{t,z,u}:=\Law\{g_t(z,u,\varepsilon_t)\mid Z=z,U=u\},
\qquad
f_{t,z,u}:=\text{density of }\kappa_{t,z,u}.
\]
Also write
\[
\Pi_z:=\Law(U\mid Z=z),
\qquad
\Pi^{\mathrm{obs}}_{t,z}:=\Law(U\mid T=t,Z=z).
\]
Thus \(\Pi_z\) is the latent mixing law under the intervention, while \(\Pi^{\mathrm{obs}}_{t,z}\) is the latent mixing law among factual units in arm \(t\).

\begin{lemma}[SEM mixture representation]
	\label{lem:sem-mixture-representation}
	Under the SEM assumptions above, whenever the conditional laws are dominated by the common reference measure used earlier in the paper,
	\[
	f_{t,z}(y)=\int f_{t,z,u}(y)\,d\Pi_z(u),
	\qquad
	f_{t,t,z}(y)=\int f_{t,z,u}(y)\,d\Pi^{\mathrm{obs}}_{t,z}(u).
	\]
	Thus the difference between \(P_{t,z}\) and \(P_{t,t,z}\) is a difference in the mixing law over the same latent-regime kernels.
\end{lemma}

\begin{proof}
	By the law of total probability conditional on \(Z=z\),
	\[
	P_{t,z}(B)=\int \kappa_{t,z,u}(B)\,d\Pi_z(u)
	\]
	for every measurable \(B\subseteq\mathcal Y\).  Similarly, conditional on \((T=t,Z=z)\), the exogenous-noise assumption gives
	\[
	\Law(Y(t)\mid T=t,Z=z,U=u)=\kappa_{t,z,u},
	\]
	so
	\[
	P_{t,t,z}(B)=\int \kappa_{t,z,u}(B)\,d\Pi^{\mathrm{obs}}_{t,z}(u).
	\]
	Taking Radon--Nikodym derivatives with respect to the common reference measure gives the two displayed density formulas.
\end{proof}

For any probability law \(\Pi\) on the latent space, define the corresponding SEM mixture density
\[
q^{t,z}_{\Pi}(y):=\int f_{t,z,u}(y)\,d\Pi(u).
\]
Then Lemma~\ref{lem:sem-mixture-representation} says
\[
q^{t,z}_{\Pi_z}=f_{t,z},
\qquad
q^{t,z}_{\Pi^{\mathrm{obs}}_{t,z}}=f_{t,t,z}.
\]
For a density level \(c>0\), write
\[
A^{t,z}_{c}(\Pi):=\{y\in\mathcal Y:q^{t,z}_{\Pi}(y)\ge c\}.
\]
The paper's topological summaries are applied to these density-induced
superlevel sets, or to a finite collection of them.

\begin{remark}[Choice of filtration]
	The choice of filtration is part of the definition of the topological summary
	\(\Psi\).  In this paper we focus on density superlevel filtrations because
	the SEM mixture representation gives explicit densities and because ordinary
	supports may be topologically uninformative in continuous models.  Other
	filtrations may be used, but they correspond to different choices of
	\(\Psi\), and hence to different causal topological estimands.
\end{remark}

\begin{remark}[How hidden confounding appears here]
	If \(T\) depends on \(U\), then generally \(\Pi^{\mathrm{obs}}_{t,z}\neq \Pi_z\), so weak ignorability fails unless the latent kernels happen to average to the same law.  Conditional topological ignorability only asks whether the reweighting \(\Pi_z\mapsto \Pi^{\mathrm{obs}}_{t,z}\) changes the selected topological signature of the density.  This is a geometric question about the mixture density, not a claim that the full law is unchanged.
\end{remark}

\subsection{A general chamber criterion}
\label{subsec:sem-chamber-criterion}

Let \(\Lambda=\{c_1>\cdots>c_m\}\) be the finite set of density levels used by the chosen signature.  A finite set is enough for the finite Betti curves, persistence summaries, Euler curves, and finite ECT signatures used in applications.  The same statements apply to a continuum of levels if all compatibility assumptions are imposed for every level.

\begin{definition}[Filtered topological chamber]
	\label{def:filtered-chamber}
	Fix \((t,z)\) and \(\Lambda=\{c_1>\cdots>c_m\}\).  
	Two latent mixing laws \(\Pi\) and \(\Pi'\) are in the same
	\(\Lambda\)-filtered topological chamber if there exist homotopy equivalences
	\[
	h_c:A^{t,z}_{c}(\Pi)\longrightarrow A^{t,z}_{c}(\Pi'),
	\qquad c\in\Lambda,
	\]
	such that, for every \(c_i\ge c_j\), the diagram
	\[
	\begin{CD}
		A^{t,z}_{c_i}(\Pi) @>>> A^{t,z}_{c_j}(\Pi) \\
		@V{h_{c_i}}VV             @VV{h_{c_j}}V \\
		A^{t,z}_{c_i}(\Pi') @>>> A^{t,z}_{c_j}(\Pi')
	\end{CD}
	\]
	commutes up to homotopy, where the horizontal maps are the canonical inclusions.
	Equivalently, the two restricted superlevel filtrations determine isomorphic persistence modules over the ordered set \(\Lambda\).
\end{definition}

\begin{figure}
	\centering
	\includegraphics[scale = 0.2]{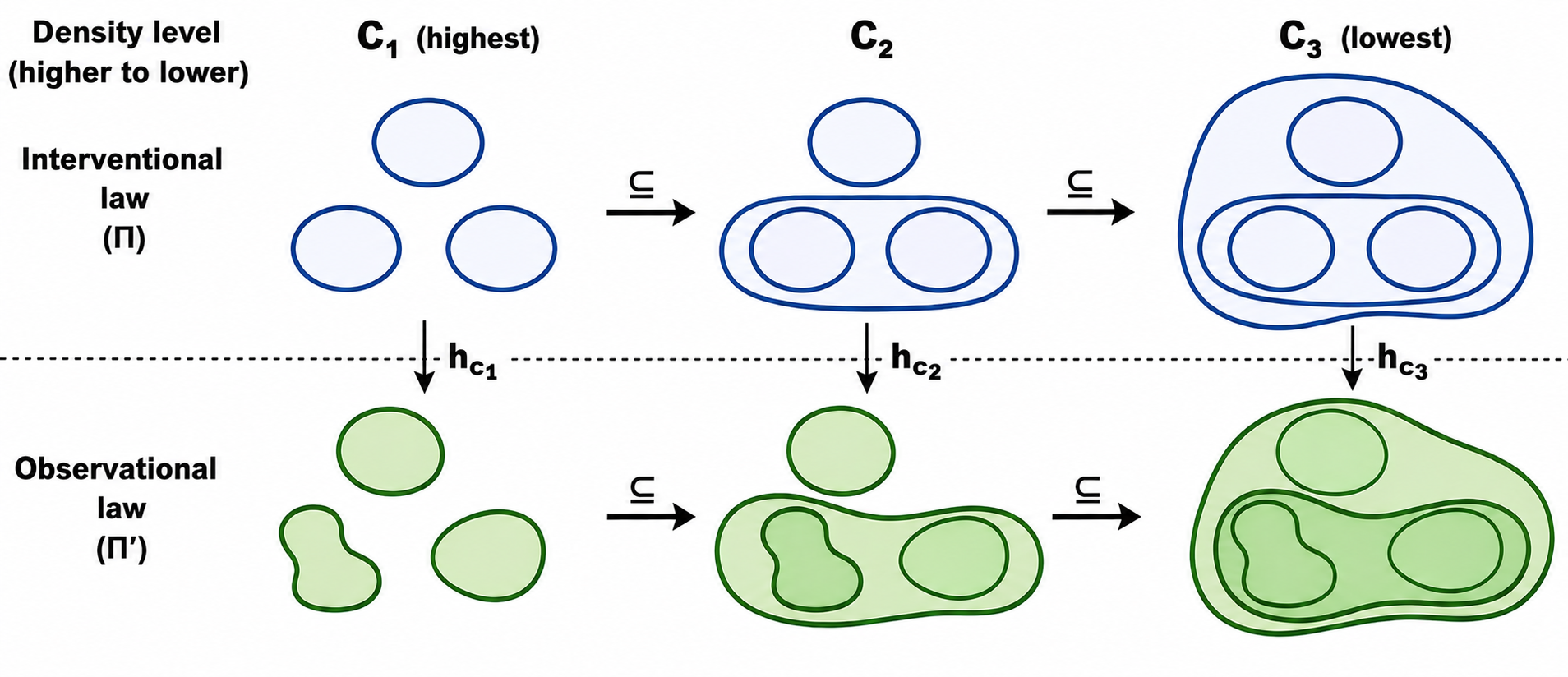}
\caption{Illustration of two filtrations lying in the same $\Lambda$-filtered topological chamber. The top row corresponds to the interventional mixing law $\Pi$, and the bottom row to the reweighted factual mixing law $\Pi'$. The maps $h_{c_i}$ denote homotopy equivalences compatible with the inclusion maps across levels, yielding isomorphic persistence modules over $\Lambda=\{C_1>C_2>C_3\}$.}\label{chamber-def}
\end{figure}
\begin{theorem}[Topology preservation under reweighting]
	\label{thm:sem-topo}
	Fix \((t,z)\).  Let \(\Psi(\mu)=\Phi(f_\mu)\) be any topological summary that depends only on the \(\Lambda\)-filtered homotopy type of the density superlevel filtration.  This includes Betti numbers at levels in \(\Lambda\), Betti curves on \(\Lambda\), Euler curves on \(\Lambda\), and finite persistent-homology summaries whose birth and death coordinates are indexed by \(\Lambda\).
	
	If \(\Pi_z\) and \(\Pi^{\mathrm{obs}}_{t,z}\) lie in the same \(\Lambda\)-filtered topological chamber, then
	\[
	\Psi(P_{t,z})=\Psi(P_{t,t,z}).
	\]
	Consequently, if this chamber condition holds for \(t=0,1\) and for \(P_Z\)-almost every \(z\), then conditional topological ignorability relative to \(\Psi\) holds.
\end{theorem}

\begin{proof}
	By Lemma~\ref{lem:sem-mixture-representation},
	\[
	f_{t,z}=q^{t,z}_{\Pi_z},
	\qquad
	f_{t,t,z}=q^{t,z}_{\Pi^{\mathrm{obs}}_{t,z}}.
	\]
	Hence the corresponding superlevel filtrations are
	\[
	\bigl(A^{t,z}_{c}(\Pi_z)\bigr)_{c\in\Lambda}
	\quad\text{and}\quad
	\bigl(A^{t,z}_{c}(\Pi^{\mathrm{obs}}_{t,z})\bigr)_{c\in\Lambda}.
	\]
	The chamber condition gives compatible homotopy equivalences between these two filtrations. Therefore they have the same filtered homotopy type and induce isomorphic persistence modules over any fixed coefficient field. Since \(\Psi\) depends only on this filtered homotopy type,
	\[
	\Psi(P_{t,z})=\Psi(P_{t,t,z}).
	\]
	If this holds for both treatment arms and for \(P_Z\)-almost every \(z\), this is precisely conditional topological ignorability relative to \(\Psi\).
\end{proof}
Theorem~\ref{thm:sem-topo} is the general principle: latent reweighting is harmless for \(\Psi\) when it does not cross a topological wall.  The following results give increasingly concrete mechanisms by which this chamber condition can hold.

\begin{corollary}[Finite-cell no-switching criterion]
	\label{cor:sem-no-active-cell-switching}
	Fix \((t,z)\).  Suppose \(\mathcal Y\) is represented by a finite cell complex \(\mathcal K\), and suppose each component density is represented by a cell-constant function.  Let \(a_\sigma(u)\) be the value of \(f_{t,z,u}\) on the relative interior of cell \(\sigma\).  For a latent mixing law \(\Pi\), set
	\[
	a_\sigma(\Pi):=\int a_\sigma(u)\,d\Pi(u).
	\]
	For \(c\in\Lambda\), define the cell superlevel set
	\[
	A^{t,z,\mathcal K}_{c}(\Pi)
	:=
	\overline{\bigcup\{\sigma^\circ:a_\sigma(\Pi)\ge c\}}.
	\]
	If, for every cell \(\sigma\in\mathcal K\) and every \(c\in\Lambda\),
	\[
	a_\sigma(\Pi_z)\ge c
	\quad\Longleftrightarrow\quad
	a_\sigma(\Pi^{\mathrm{obs}}_{t,z})\ge c,
	\]
	then the two finite filtrations are identical. Hence
\[
\Psi(P_{t,z})=\Psi(P_{t,t,z})
\]
for every \(\Psi\) computed from that finite cell filtration.
\end{corollary}

\begin{proof}
	The displayed equivalence says that the active cells are exactly the same for \(\Pi_z\) and \(\Pi^{\mathrm{obs}}_{t,z}\) at every level \(c\in\Lambda\).  Therefore the cell superlevel sets, and hence the whole finite filtration, are identical.  All summaries computed from that filtration are identical.
\end{proof}

\begin{remark}[The chamber is a set of linear inequalities]
	\label{rem:linear-chamber-inequalities}
	If \(U\in\{1,\ldots,r\}\), write \(w_i=\Pi(U=i)\).  Then
	\[
	a_\sigma(\Pi)=\sum_{i=1}^r w_i a_\sigma(i).
	\]
	For each pair \((\sigma,c)\), the wall
	\[
	\sum_{i=1}^r w_i a_\sigma(i)=c
	\]
	is a hyperplane in the probability simplex.  The no-switching condition says that the interventional weights and the factual weights lie on the same side of every relevant hyperplane.  Thus hidden confounding may change the weights substantially, but it is harmless for the selected finite topological signature as long as the change remains inside the same chamber (see Figure \ref{chamber-wall}).
\end{remark}
\begin{figure}
\centering
\includegraphics[scale =0.2]{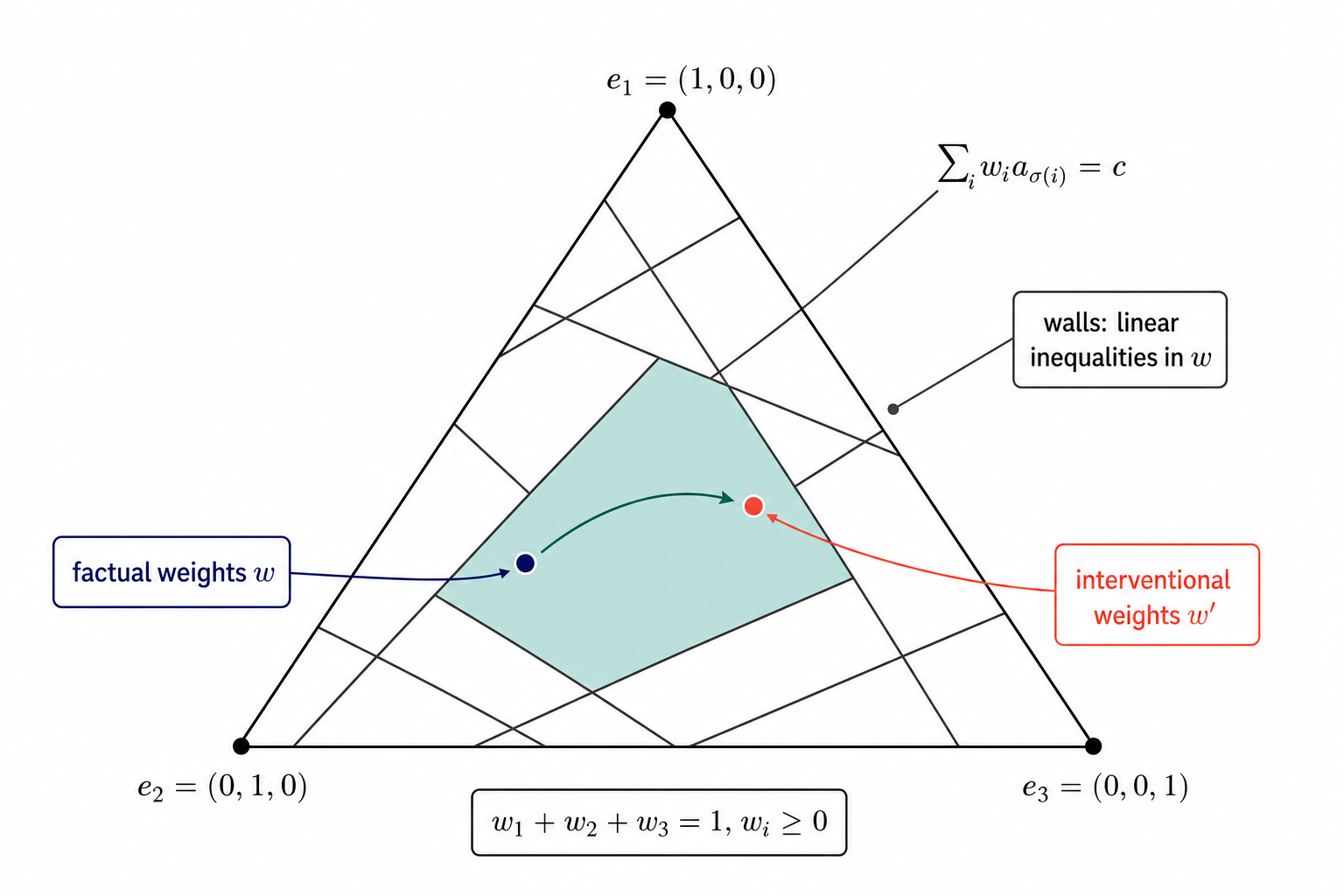}
\caption{
	Hyperplane walls in the probability simplex.  
	Changing between interventional and factual weights moves the latent mixture weights \(w_i=\Pi(U=i)\).  
	The finite-cell topology is preserved as long as the path does not cross any wall
	\(\sum_i w_i a_\sigma(i)=c\).
}\label{chamber-wall}
\end{figure}

\begin{theorem}[Smooth no-bifurcation criterion]
	\label{thm:sem-no-bifurcation}
	Fix \((t,z)\), and let
	\[
	\Pi_s=(1-s)\Pi_z+s\Pi^{\mathrm{obs}}_{t,z},
	\qquad 0\le s\le 1,
	\]
	with \(q_s=q^{t,z}_{\Pi_s}\).  Assume that the relevant level sets are compact and that \((y,s)\mapsto q_s(y)\) is \(C^2\).  Suppose that for every \(c\in\Lambda\), \(c\) is a regular value of \(q_s\) for every \(s\in[0,1]\); equivalently,
	\[
	q_s(y)=c \quad\Longrightarrow\quad \nabla_y q_s(y)\neq 0.
	\]
	Then, for each \(c\in\Lambda\), the superlevel sets
	\[
	A^{t,z}_{c}(\Pi_z)
	\quad\text{and}\quad
	A^{t,z}_{c}(\Pi^{\mathrm{obs}}_{t,z})
	\]
	are diffeomorphic.  If the induced diffeomorphisms are chosen compatibly over \(\Lambda\), then the two latent mixing laws lie in the same \(\Lambda\)-filtered topological chamber, and Theorem~\ref{thm:sem-topo} applies.
\end{theorem}

\begin{proof}
	Fix \(c\in\Lambda\). Since \((y,s)\mapsto q_s(y)\) is \(C^2\) and \(c\) is a regular value of \(q_s\) for every \(s\in[0,1]\), the boundaries
	\[
	\{y:q_s(y)=c\}
	\]
	remain nonsingular hypersurfaces along the path. Together with compactness, the standard deformation principle in Morse theory implies that no topological bifurcation occurs at level \(c\) (see Milnor~\cite{MilnorMorse}). Equivalently, the superlevel sets
	\[
	A_c(s)=\{y:q_s(y)\ge c\}
	\]
	form a smooth isotopy as \(s\) varies; see, e.g., Guillemin and Pollack~\cite{GuilleminPollack}. Hence \(A_c(0)\) and \(A_c(1)\) are diffeomorphic. Applying this to every \(c\in\Lambda\), and assuming compatibility with inclusions across levels, gives equality of the filtered homotopy type. The conclusion follows from Theorem~\ref{thm:sem-topo}.
\end{proof}

\subsection{Component and skeleton criteria}
For \(H_0\)-based summaries, preservation means that connected components and their merge history are unchanged across the selected levels. Equivalently, the interventional and factual superlevel filtrations must have isomorphic merge trees over \(\Lambda\). In that case, the persistent \(H_0\) modules and finite \(H_0\) barcodes agree.

For loops and higher-dimensional holes, one can either work with nerves of good covers or use the stronger condition that both filtrations deformation retract onto a common skeleton. The formal nerve-based criterion is given in Appendix~\ref{app:nerve-criterion}. Here we state the common-skeleton condition because it is geometrically transparent and is illustrated in Figure~\ref{deformation}.

\begin{theorem}[Common deformation-retract skeleton]
	\label{thm:sem-common-skeleton}
	Fix \((t,z)\) and \(\Lambda\).  Suppose that, for every \(c\in\Lambda\), there exists a space \(M_c\) and deformation retractions
	\[
	r_c:A^{t,z}_{c}(\Pi_z)\to M_c,
	\qquad
	r_c^{\mathrm{obs}}:A^{t,z}_{c}(\Pi^{\mathrm{obs}}_{t,z})\to M_c,
	\]
	and suppose these deformation retractions commute with the inclusion maps across levels up to homotopy.   Then the factual and interventional filtrations have isomorphic persistent homology in every dimension over \(\Lambda\).  Consequently, all Betti curves, finite barcodes, and Euler characteristics over \(\Lambda\) are preserved.
\end{theorem}

\begin{proof}
	A deformation retract is a homotopy equivalence.  Hence both level sets have homology isomorphic to \(H_k(M_c;\mathbb F)\) at every level and in every dimension.  Compatibility across levels makes these isomorphisms commute with the persistence maps.  Therefore the persistence modules are isomorphic.  Euler characteristic is preserved because, for finite CW complexes, \(\chi=\sum_k(-1)^k\beta_k\).
\end{proof}

\begin{figure}
	\centering
	\includegraphics[scale =0.2]{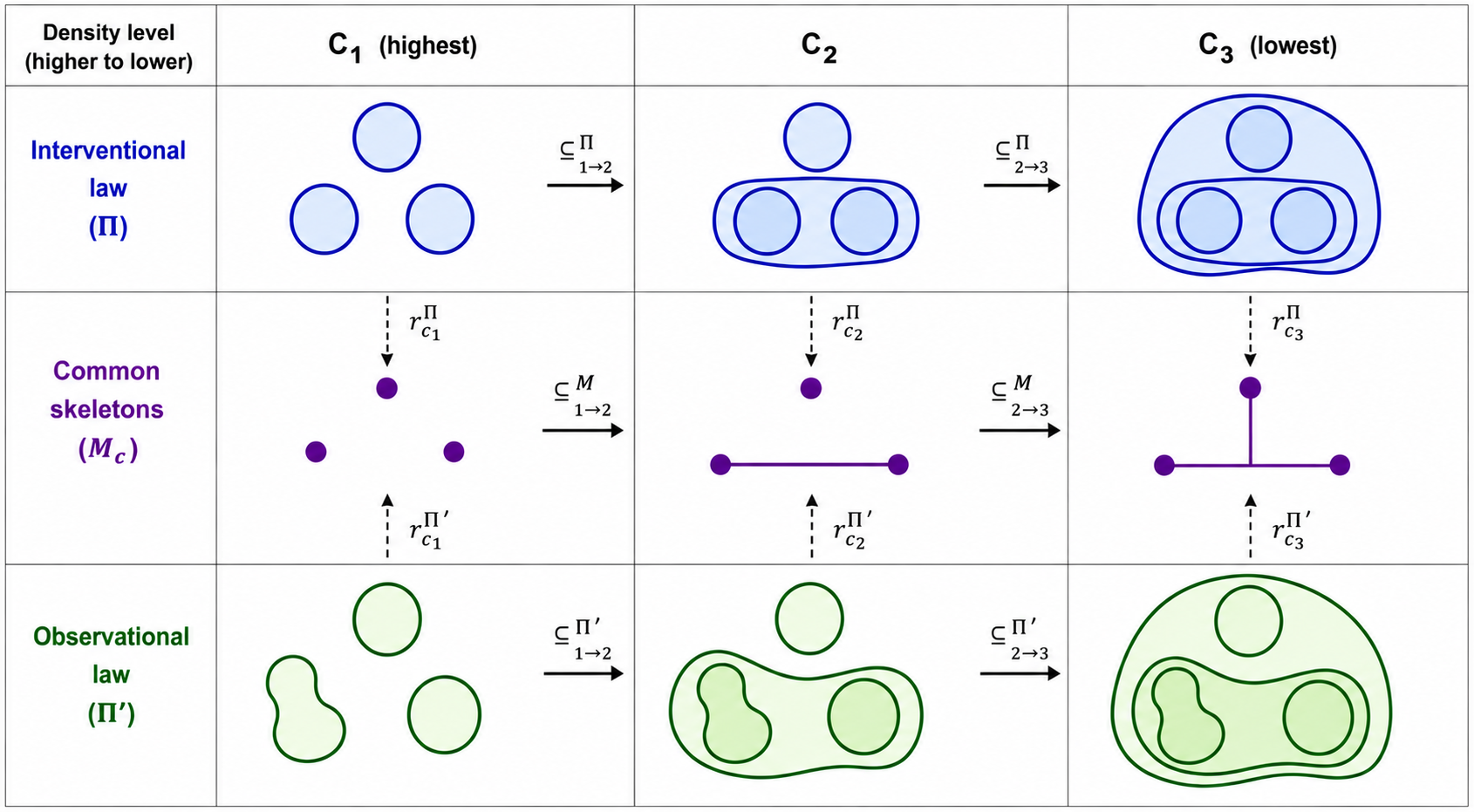}
\caption{
	Illustration of the common deformation-retract skeleton criterion.  
	At each density level \(c_i\), the interventional and factual superlevel sets deformation retract onto the same skeleton \(M_{c_i}\).  
	The inclusion maps between levels commute with the deformation retractions, yielding isomorphic persistent homology across the two filtrations.
}\label{deformation}
\end{figure}
\subsection{Euler characteristic and finite ECT summaries}
\label{subsec:sem-euler-ect}
Euler characteristic is coarser than persistent homology: it may remain unchanged even when individual Betti numbers change, because alternating contributions can cancel. The following criterion captures this weaker form of preservation.

\begin{theorem}[Euler cancellation criterion \cite{Hatcher2002}]
	\label{thm:sem-euler-cancellation}
	Let \(S\) and \(S'\) be finite CW complexes with decompositions
	\[
	S=C\cup B,
	\qquad
	S'=C\cup B',
	\]
	where all displayed sets and intersections are finite CW subcomplexes.  If
	\[
	\chi(B)-\chi(C\cap B)
	=
	\chi(B')-\chi(C\cap B'),
	\]
	then
	\[
	\chi(S)=\chi(S').
	\]
	In particular, if \(S'=C\subseteq S=C\cup B\) and \(\chi(B)=\chi(C\cap B)\), then \(\chi(S)=\chi(S')\).
\end{theorem}

\begin{proof}
	Euler characteristic is additive over finite CW unions:
	\[
	\chi(C\cup B)=\chi(C)+\chi(B)-\chi(C\cap B),
	\]
	and similarly for \(C\cup B'\).  The equality of the two increments gives \(\chi(S)=\chi(S')\).  The special case follows by taking \(B'=\emptyset\).
\end{proof}
For example, adding a contractible appendage along a contractible overlap changes the geometry but contributes zero Euler increment.
\begin{corollary}[Euler-curve conditional topological ignorability]
	\label{cor:sem-euler-curve-ignorability}
	Let
	\[
	\Psi_\chi(\mu)=\bigl(\chi(\{y:f_\mu(y)\ge c\})\bigr)_{c\in\Lambda}.
	\]
	If, for each \(t\in\{0,1\}\), for \(P_Z\)-almost every \(z\), and for every \(c\in\Lambda\), the pair
	\[
	A^{t,z}_{c}(\Pi_z),
	\qquad
	A^{t,z}_{c}(\Pi^{\mathrm{obs}}_{t,z})
	\]
	satisfies the Euler cancellation criterion, then conditional topological ignorability holds relative to \(\Psi_\chi\).
\end{corollary}

\begin{proof}
	The Euler cancellation criterion gives equality of Euler characteristic at every selected density level.  Therefore the Euler-curve vectors for \(P_{t,z}\) and \(P_{t,t,z}\) are equal for every relevant \((t,z)\).
\end{proof}

For location-sensitive summaries, the paper also considers Euler-transform-type signatures.  For a compact tame set \(S\subseteq\mathbb R^d\), define
\[
\operatorname{ECT}_S(v,a):=\chi(S\cap H_{v,a}),
\qquad
H_{v,a}:=\{y\in\mathbb R^d:y\cdot v\le a\},
\]
where \(v\in S^{d-1}\) and \(a\in\mathbb R\).  In computations, \((v,a)\) is usually restricted to a finite grid.

The same cancellation idea applies slice-by-slice to finite ECT signatures. Indeed, each finite ECT value is an Euler characteristic of a sliced set \(S\cap H_{v,a}\). Therefore, if the Euler cancellation condition holds for every selected slice and every selected density level, the finite ECT signature is preserved under reweighting.
\subsection{Linear Gaussian SEMs and topological bifurcations}
\label{subsec:linear-gaussian-bifurcations}
We now interpret conditional topological ignorability in a concrete and geometrically intuitive SEM family.  
In the linear Gaussian SEM below, the component means and covariance are fixed within each treatment arm, while hidden confounding changes only the relative weights of the components.  Thus treatment selection changes the prominence of existing component regions without moving their centers, and topological changes occur through geometric bifurcations of the corresponding superlevel sets.

Let \(U\in\{1,\ldots,r\}\), and define
\[
p_i(z):=\Prob(U=i\mid Z=z),
\qquad
\sum_{i=1}^r p_i(z)=1.
\]
Assume treatment assignment probabilities
\[
e_{t,i}(z):=\Prob(T=t\mid Z=z,U=i),
\qquad t\in\{0,1\}.
\]
Bayes' rule gives the factual latent weights
\[
p^{\mathrm{obs}}_{t,i}(z)
=
\Prob(U=i\mid T=t,Z=z)
=
\frac{e_{t,i}(z)p_i(z)}
{\sum_{j=1}^r e_{t,j}(z)p_j(z)}.
\]
Consider the linear Gaussian SEM
\[
Y(t)=\alpha_t+B_t z+\Gamma_t u_i+\varepsilon_t,
\qquad
\varepsilon_t\sim N(0,\Sigma_t),
\]
where \(\Sigma_t\) is positive definite and \(\varepsilon_t\) is conditionally independent of \((T,U)\) given \(Z\).  Then
\[
f_{t,z}(y)
=
\sum_{i=1}^r p_i(z)\,
\varphi_d(y;\mu_{t,z,i},\Sigma_t),
\]
and
\[
f_{t,t,z}(y)
=
\sum_{i=1}^r p^{\mathrm{obs}}_{t,i}(z)\,
\varphi_d(y;\mu_{t,z,i},\Sigma_t),
\]
where
\[
\mu_{t,z,i}
=
\alpha_t+B_t z+\Gamma_t u_i.
\]

The interventional and factual laws therefore differ only through the mixture weights.  Geometrically, treatment selection changes the prominence of the component regions without moving their centers.

For \(w=(w_1,\ldots,w_r)\in\Delta_r\), define
\[
q_w(y)
:=
\sum_{i=1}^r
w_i\,
\varphi_d(y;\mu_{t,z,i},\Sigma_t).
\]
The interventional and factual laws correspond to two points
\[
w=(p_i(z))_{i=1}^r,
\qquad
w'=(p^{\mathrm{obs}}_{t,i}(z))_{i=1}^r
\]
inside the probability simplex \(\Delta_r\).

Fix density levels \(\Lambda\).  Define the Gaussian bifurcation set
\[
\mathcal B_{t,z,\Lambda}
:=
\Bigl\{
w\in\Delta_r:
\exists c\in\Lambda,\ \exists y
\text{ such that }
q_w(y)=c
\text{ and }
\nabla_y q_w(y)=0
\Bigr\}.
\]
This is the set of weights at which a selected density level undergoes a topological transition.

In Gaussian mixtures, crossing the bifurcation set has a direct geometric interpretation (see Figure~\ref{mixure}):
\begin{itemize}
	\item a connected component may disappear when one component region falls below the chosen density threshold;
	\item two modal regions may merge through a saddle connection;
	\item a loop surrounding a low-density region may open or collapse.
\end{itemize}
\begin{figure}
	\centering
	\includegraphics[scale = 0.22]{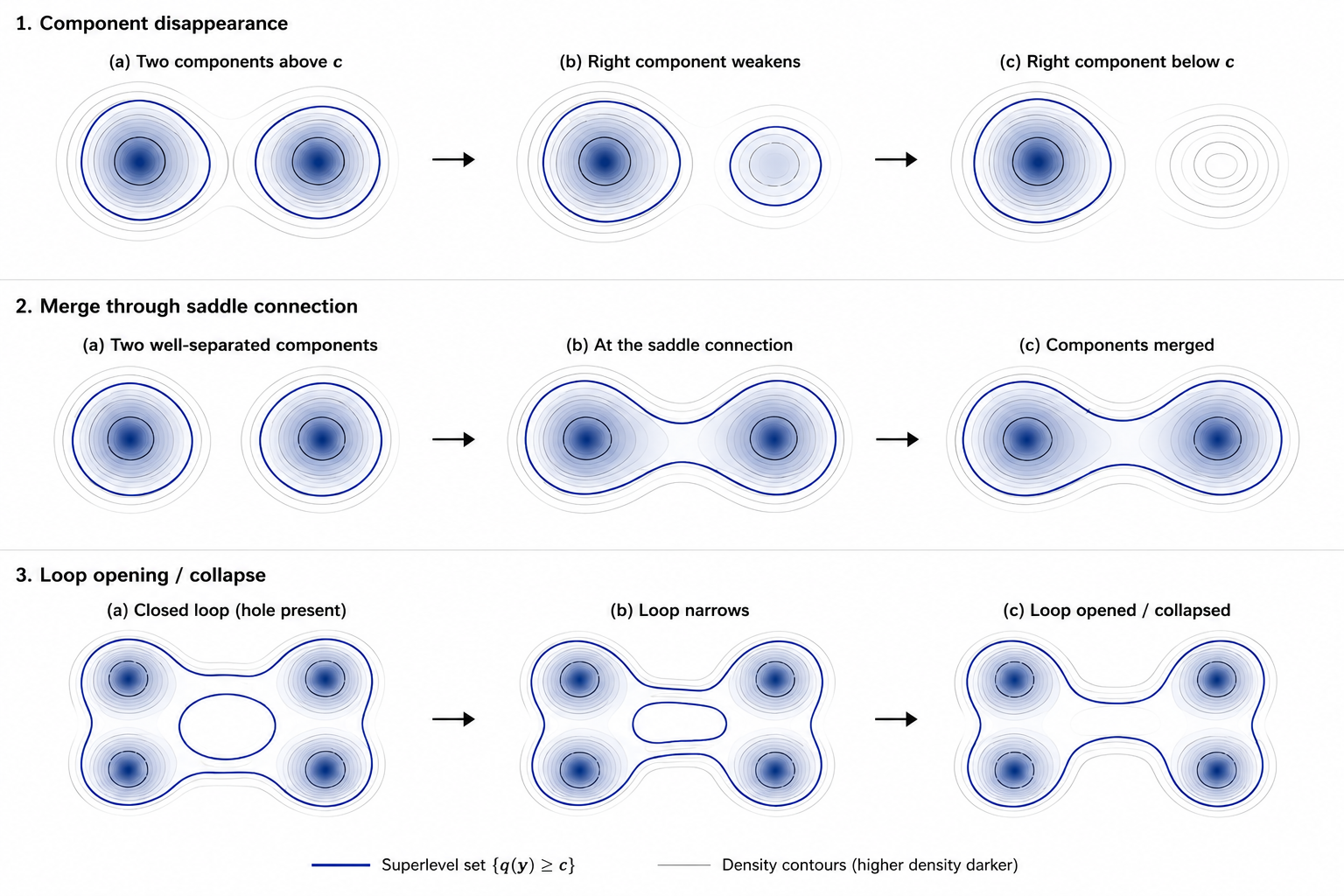}
\caption{
	Topological bifurcations of Gaussian-mixture superlevel sets: component disappearance, saddle merging, and loop collapse under mixture-weight reweighting.
}\label{mixure}
\end{figure}

Thus conditional topological ignorability becomes a geometric wall-crossing problem inside the probability simplex.  The posterior-weight trajectory
\[
(1-s)w+sw',
\qquad 0\le s\le1,
\]
may substantially change the full density and therefore violate weak ignorability, while still preserving the selected topological signature provided it remains inside a single chamber of
\[
\Delta_r\setminus\mathcal B_{t,z,\Lambda}.
\]

\begin{remark}
	The significance of the Gaussian SEM is therefore not merely smoothness.  
	Rather, Gaussian mixtures provide a concrete geometric realization of topological ignorability:
	hidden confounding induces trajectories in mixture-weight space, while topological summaries change only when those trajectories cross bifurcation walls corresponding to component births, component deaths, saddle merges, or loop collapses.
\end{remark}

Overall, the SEM analysis shows that conditional topological ignorability should be understood as a stability property of selected topological summaries under latent reweighting.  The criteria above correspond to different levels of structure: exact cell-wise no-switching, smooth absence of bifurcations, preservation of component merge trees, common deformation-retract skeletons, and Euler/ECT cancellation.  These criteria are not meant as directly observable tests from a single observational distribution, but as structural mechanisms explaining how topological causal summaries can remain identifiable even when the full interventional law is not.

\section{A worked hidden-confounding example}
\label{sec:worked-example}

We now give a fully explicit example showing that hidden confounding may change
the full probability law while preserving a selected superlevel-set topology.
The example is not intended to verify conditional topological ignorability from
observed data alone. Rather, it shows that the SEM mechanisms described above
can occur in an exactly computable model.

Let \(U\in\{0,1\}\) be an unobserved variable with
\[
\Prob(U=0)=\Prob(U=1)=\frac12.
\]
Treatment is assigned by
\[
T=\mathbf 1\{V\le p(U)\},
\qquad
p(0)=\frac15,\quad p(1)=\frac45,
\]
where \(V\sim\mathrm{Unif}(0,1)\) is independent of \(U\). Thus \(U\)
affects treatment assignment.

The potential-outcome maps, given explicitly in Appendix~\ref{app:worked-sem}, are chosen so that the interventional densities are
\[
f_0(y)
=
\frac{4}{15}\mathbf 1_{[0,1]}(y)
+
\frac{19}{60}\mathbf 1_{[1,2]}(y)
+
\frac{5}{12}\mathbf 1_{[2,3]}(y),
\]
and
\[
f_1(y)
=
\frac{7}{20}\mathbf 1_{[0,1]}(y)
+
\frac{13}{20}\mathbf 1_{[2,3]}(y).
\]
By Bayes' rule,
\[
\Prob(U=0\mid T=1)=\frac15,
\qquad
\Prob(U=1\mid T=1)=\frac45,
\]
and
\[
\Prob(U=0\mid T=0)=\frac45,
\qquad
\Prob(U=1\mid T=0)=\frac15.
\]
The resulting factual densities are
\[
h_0(y)
=
\frac{23}{75}\mathbf 1_{[0,1]}(y)
+
\frac{49}{150}\mathbf 1_{[1,2]}(y)
+
\frac{11}{30}\mathbf 1_{[2,3]}(y),
\]
and
\[
h_1(y)
=
\frac{13}{50}\mathbf 1_{[0,1]}(y)
+
\frac{37}{50}\mathbf 1_{[2,3]}(y).
\]
The appendix also verifies the branchwise density calculations.

Let
\[
\mu_a:=\Law(Y(a)),
\qquad
\nu_a:=\Law(Y\mid T=a),
\]
so that \(f_a\) is the density of the interventional law \(\mu_a\), while
\(h_a\) is the density of the factual law \(\nu_a\).

Fix the density level
\[
\Lambda=\{c\},
\qquad
c=0.20,
\]
and define, for any law \(\mu\) with density \(f_\mu\),
\[
\Psi_\Lambda(\mu)
=
\beta_0(\{y:f_\mu(y)\ge c\}).
\]
This is a one-level superlevel-set topological summary of the type used in the
SEM criteria above. For the displayed piecewise-constant densities, superlevel
sets are understood using the natural closed representatives; equivalently, one
may take closures of the corresponding superlevel regions, so endpoint
conventions do not affect the connected-component counts.

\begin{table}[htbp]
	\centering
	\small
	\renewcommand{\arraystretch}{1.15}
	\caption{Superlevel sets at \(c=0.20\). Hidden confounding changes the density values but not the selected superlevel-set topology.}
	\label{tab:superlevel-example}
	\begin{tabular}{cccc}
		\toprule
		Quantity & Density & Superlevel set at \(c=0.20\) & \(\beta_0\)\\
		\midrule
		Interventional control \(f_0\)
		& \(\frac{4}{15}\mathbf 1_{[0,1]}+\frac{19}{60}\mathbf 1_{[1,2]}+\frac{5}{12}\mathbf 1_{[2,3]}\)
		& \([0,3]\)
		& \(1\) \\
		Factual control \(h_0\)
		& \(\frac{23}{75}\mathbf 1_{[0,1]}+\frac{49}{150}\mathbf 1_{[1,2]}+\frac{11}{30}\mathbf 1_{[2,3]}\)
		& \([0,3]\)
		& \(1\) \\
		Interventional treatment \(f_1\)
		& \(\frac{7}{20}\mathbf 1_{[0,1]}+\frac{13}{20}\mathbf 1_{[2,3]}\)
		& \([0,1]\cup[2,3]\)
		& \(2\) \\
		Factual treatment \(h_1\)
		& \(\frac{13}{50}\mathbf 1_{[0,1]}+\frac{37}{50}\mathbf 1_{[2,3]}\)
		& \([0,1]\cup[2,3]\)
		& \(2\) \\
		\bottomrule
	\end{tabular}
\end{table}
\begin{proposition}[Hidden confounding with preserved superlevel topology]
	\label{prop:superlevel-example}
In the SEM above, weak ignorability fails, but the selected superlevel-set
topological summary is preserved under factual reweighting:
	\[
	\Psi_\Lambda(\mu_0)=\Psi_\Lambda(\nu_0)=1,
	\qquad
	\Psi_\Lambda(\mu_1)=\Psi_\Lambda(\nu_1)=2.
	\]
	Thus the selected topological treatment contrast equals
	\[
	\Psi_\Lambda(\mu_1)-\Psi_\Lambda(\mu_0)=1,
	\]
	and the same value is obtained from the factual laws:
	\[
	\Psi_\Lambda(\nu_1)-\Psi_\Lambda(\nu_0)=1.
	\]
	Nevertheless, weak ignorability fails because
	\[
	\mu_1\neq \nu_1.
	\]
\end{proposition}
\begin{proof}
	Weak ignorability fails because the law of \(Y(1)\) differs from the factual
	treated law. For example,
	\[
	\Prob(Y(1)\in[0,1])=\int_0^1 f_1(y)\,dy=\frac{7}{20},
	\]
	whereas
	\[
	\Prob(Y\in[0,1]\mid T=1)=\int_0^1 h_1(y)\,dy=\frac{13}{50}.
	\]
	Hence \(\mu_1\neq\nu_1\).
	
	For the topological statement, Table~\ref{tab:superlevel-example} gives
	\[
	\{y:f_0(y)\ge0.20\}
	=
	\{y:h_0(y)\ge0.20\}
	=
	[0,3],
	\]
	and
	\[
	\{y:f_1(y)\ge0.20\}
	=
	\{y:h_1(y)\ge0.20\}
	=
	[0,1]\cup[2,3].
	\]
	Therefore the corresponding numbers of connected components are \(1\) in the
	control arm and \(2\) in the treatment arm for both the interventional and
	factual laws. This proves the claim.
\end{proof}
This example illustrates the central distinction of the paper. Hidden
confounding changes the full density, so classical weak ignorability fails.
However, for the selected density level \(c=0.20\), the superlevel-set
connectivity is unchanged. The identified object is therefore not the full law
or the mean effect, but a coarser topological contrast of high-density regions.

 Appendix~\ref{app:same-observational-different-ate} further gives a second SEM with the same observational law but a different ATE, emphasizing that the mean effect is not identified on this confounded model class.

\FloatBarrier

\section{Estimation under conditional topological ignorability}

The identification results show that the estimand depends on the transform. Under weak ignorability, the full interventional conditional laws are identified from observed conditional laws, so one can target transformed marginal laws and hence \(\Theta_{\Psi}\).  Under conditional topological ignorability with a noninjective \(\Psi\), the naturally identified target is instead the standardized contrast
\[
\tau_{\Psi}=\E\!\left[\norm{\Psi(P_{1,Z})-\Psi(P_{0,Z})}_E\right].
\]
This section develops estimation for that regime.

\subsection{A general conditional plug-in estimator}

Write
\[
\nu_{t,z}:=\Law(Y\mid T=t,Z=z).
\]
Under Proposition~\ref{prop:tau-ident},
\[
\Psi(P_{t,z})=\Psi(\nu_{t,z})\qquad\text{a.s.}
\]
Thus one does not need to recover the full interventional law $P_t$ in order to estimate $\tau_{\Psi}$. It is enough to estimate the arm-specific \emph{factual} conditional laws and then apply the chosen topological transform.

Let $\widehat \nu_{t,z}$ be an estimator of $\nu_{t,z}$ obtained from any arm-specific nuisance procedure: conditional density estimation, localized empirical measures, finite mixtures, nearest-neighbor neighborhoods, or another conditional law estimator appropriate to the outcome type. In practice, \(\widehat\nu_{t,z}\) may be obtained with sample splitting or cross-fitting to reduce overfitting bias when flexible nuisance estimators are used. Define the estimated topological signature
\[
\widehat S_t(z):=\Psi(\widehat \nu_{t,z}),
\qquad
S_t(z):=\Psi(\nu_{t,z}),
\]
the estimated stratum-specific contrast
\[
\widehat\theta_{\Psi}(z):=\norm{\widehat S_1(z)-\widehat S_0(z)}_E,
\]
and the empirical standardized estimator
\[
\widehat\tau_{\Psi}:=\frac1n\sum_{i=1}^n \widehat\theta_{\Psi}(Z_i).
\]
The following result gives a simple sufficient condition for consistency. It is intentionally stated at the level of the estimated signatures, rather than at the level of density estimators.
\begin{proposition}[Consistency of the standardized plug-in estimator]
\label{prop:plugin-tau}
Assume conditional topological ignorability and positivity. Suppose
\[
\sup_{z\in\mathcal Z}\norm{\widehat S_t(z)-S_t(z)}_E\xrightarrow{P}0,\qquad t=0,1,
\]
and $\E[\theta_{\Psi}(Z)]<\infty$. Then
\[
\widehat\tau_{\Psi}\xrightarrow{P}\tau_{\Psi}.
\]
\end{proposition}

\begin{proof}
By the reverse triangle inequality,
\[
\left|\widehat\theta_{\Psi}(z)-\theta_{\Psi}(z)\right|
\le \norm{\widehat S_1(z)-S_1(z)}_E+\norm{\widehat S_0(z)-S_0(z)}_E.
\]
Averaging over the sample gives
\[
\left|\widehat\tau_{\Psi}-\frac1n\sum_{i=1}^n \theta_{\Psi}(Z_i)\right|
\le \sup_z \norm{\widehat S_1(z)-S_1(z)}_E+\sup_z \norm{\widehat S_0(z)-S_0(z)}_E,
\]
which converges to zero in probability by assumption. The law of large numbers yields
\[
\frac1n\sum_{i=1}^n \theta_{\Psi}(Z_i)\xrightarrow{P}\E[\theta_{\Psi}(Z)]=\tau_{\Psi}.
\]
Combining the two displays proves the result.
\end{proof}

Proposition~\ref{prop:plugin-tau} is the main estimation result for the noninjective regime. It shows that conditional topological ignorability leads to a direct estimation principle: estimate the observed arm-specific conditional signatures and average their contrast over the covariate distribution. No appeal to weak ignorability is needed once conditional topological ignorability is taken as the identifying assumption.

\subsection{A practical balancing-bin estimator}

In many data sets \(Z\) is high-dimensional, so exact conditioning on \(Z=z\)
is not practical. A pragmatic implementation is to replace exact conditioning
by grouping units according to a low-dimensional balancing summary
\[
b(Z),
\]
such as a propensity score, a prognostic score, or a learned balancing
representation. This step should be viewed as an additional approximation:
the causal estimand is still defined through conditioning on \(Z\), while the
bin estimator approximates this conditioning by comparing treated and control
factual laws within bins of \(b(Z)\). The approximation is credible when the
selected conditional topological signatures are stable within each bin.

Let \(B_1,\ldots,B_K\) be a partition of the support of \(b(Z)\), chosen so
that each bin contains adequate treated and control observations. In practice,
\(b\) is usually unknown and is replaced by an estimator \(\widehat b\). For
example, if \(b(Z)=e(Z)=\Prob(T=1\mid Z)\), then \(\widehat b(Z)=\widehat e(Z)\)
is an estimated propensity score.

Define the empirical bin weights
\[
\widehat\pi_k
:=
\frac1n\sum_{i=1}^n
\mathbf 1\{\widehat b(Z_i)\in B_k\},
\]
and the arm-specific empirical laws
\[
\widehat\nu_{t,k}
:=
\Law_n\!\bigl(Y\mid T=t,\ \widehat b(Z)\in B_k\bigr).
\]
The binwise topological contrast is
\[
\widehat\Delta_{\Psi,k}
:=
\norm{
	\Psi(\widehat\nu_{1,k})-\Psi(\widehat\nu_{0,k})
}_E,
\]
and the balancing-bin estimator is
\[
\widehat\tau_{\Psi}^{(K)}
:=
\sum_{k=1}^K
\widehat\pi_k\,\widehat\Delta_{\Psi,k}.
\]

At the population level, the corresponding bin approximation is
\[
\tau_{\Psi}^{(K)}
:=
\sum_{k=1}^K
\pi_k
\norm{
	\Psi(\nu_{1,k})-\Psi(\nu_{0,k})
}_E,
\]
where
\[
\pi_k:=\Prob(b(Z)\in B_k),
\qquad
\nu_{t,k}:=\Law(Y\mid T=t,\ b(Z)\in B_k).
\]
This is not exactly \(\tau_\Psi\), because exact conditioning on \(Z\) has
been replaced by pooling within bins of \(b(Z)\) before applying \(\Psi\).
The approximation is accurate when the selected conditional topological
signatures are stable within each bin, so that this pooling does not change
the relevant topology. Thus \(\widehat\tau_{\Psi}^{(K)}\) should be interpreted as an estimator of
\(\tau_{\Psi}^{(K)}\), with an additional binning approximation error relative
to \(\tau_\Psi\).

\begin{center}
	\setlength{\fboxsep}{8pt}
	\fbox{\begin{minipage}{0.94\textwidth}
			\textbf{Algorithmic template under conditional topological ignorability.}
			\begin{enumerate}[leftmargin=1.4em]
				\item Estimate a low-dimensional balancing summary \(\widehat b(Z)\), such as a propensity score or a learned balancing representation.
				\item Partition the sample into bins or neighborhoods with treated and control observations in each subset.
				\item Within each subset, estimate the treated and control factual laws and compute their topological signatures.
				\item Compute the within-subset contrast \(\widehat\Delta_{\Psi,k}\).
				\item Average those contrasts over the empirical bin distribution to obtain \(\widehat\tau_{\Psi}^{(K)}\).
			\end{enumerate}
			This procedure estimates a binwise approximation to the causal topological
			effect identified under conditional topological ignorability; the approximation
			improves when the chosen signatures are stable within bins.
	\end{minipage}}
\end{center}

The estimator \(\widehat\tau_{\Psi}^{(K)}\) encodes the correct causal order of
operations for the standardized target in the noninjective setting: condition
or balance first, apply \(\Psi\) second, and average the resulting topological
contrasts third. Reversing that order generally targets a different quantity
because of Remark~\ref{rem:noncommute}.

\subsection{Finite-signature approximation}

Practical topological procedures rarely work with an infinite-dimensional transform. One computes a finite vector of Euler scan values, a finite persistence landscape grid, a persistence image, or another discretized signature. The next result separates this numerical approximation from causal identification itself.

\begin{proposition}[Finite-signature approximation error]
\label{prop:finite-grid}
Let \(\mathcal F\) denote the class of laws on which \(\Psi\) is defined, and 
let  $\Psi_m:\mathcal F\to E_m$ be a finite-rank or discretized approximation to $\Psi:\mathcal F\to E$. Suppose there exists an isometric embedding $\iota_m:E_m\to E$ and a model class $\mathcal G\subseteq\mathcal F$ such that
\[
\sup_{\nu\in\mathcal G}\norm{\iota_m\Psi_m(\nu)-\Psi(\nu)}_E\le a_m.
\]
Then for every $z$ with $\nu_{0,z},\nu_{1,z}\in\mathcal G$,
\[
\left|\norm{\Psi_m(\nu_{1,z})-\Psi_m(\nu_{0,z})}_{E_m}-\theta_{\Psi}(z)\right|\le 2a_m.
\]
Consequently,
\[
\left|\tau_{\Psi_m}-\tau_{\Psi}\right|\le 2a_m,
\qquad
\tau_{\Psi_m}:=\E\!\left[\norm{\Psi_m(\nu_{1,Z})-\Psi_m(\nu_{0,Z})}_{E_m}\right].
\]
\end{proposition}

\begin{proof}
Because $\iota_m$ is isometric,
\begin{align*}
\Bigl|\norm{\Psi_m(\nu_{1,z})-\Psi_m(\nu_{0,z})}_{E_m}-\theta_{\Psi}(z)\Bigr|
&=\Bigl|\norm{\iota_m\Psi_m(\nu_{1,z})-\iota_m\Psi_m(\nu_{0,z})}_E-\norm{\Psi(\nu_{1,z})-\Psi(\nu_{0,z})}_E\Bigr|\\
&\le \norm{\iota_m\Psi_m(\nu_{1,z})-\Psi(\nu_{1,z})}_E+\norm{\iota_m\Psi_m(\nu_{0,z})-\Psi(\nu_{0,z})}_E\\
&\le 2a_m.
\end{align*}
Taking expectations over $Z$ gives the final bound.
\end{proof}

Proposition~\ref{prop:finite-grid} controls only the numerical error caused by replacing the ideal topological transform \(\Psi\) with a finite signature \(\Psi_m\). This is distinct from the causal identification step and from the statistical estimation of the conditional factual laws \(\nu_{t,z}\). Conditional topological ignorability justifies replacing \(\Psi(P_{t,z})\) by \(\Psi(\nu_{t,z})\); nuisance estimation approximates \(\nu_{t,z}\) by \(\widehat\nu_{t,z}\); and the finite-signature approximation replaces \(\Psi\) by \(\Psi_m\). These three layers should be analyzed separately.

\subsection{Persistent homology in practice}

The estimation scheme above applies directly to persistent homology. Suppose $Y$ is multivariate and, within each treated or control subset, one computes a filtration from the observed outcome cloud, a smoothed density level set, or a distance-to-measure transform. Let $\Psi_{\mathrm{PH}}$ be a vectorized persistence summary such as a persistence landscape or persistence image. Then the binwise persistent-homology estimator is
\[
\widehat\tau_{\Psi_{\mathrm{PH}}}^{(K)}=
\sum_{k=1}^K \widehat\pi_k
\,
\norm{
\Psi_{\mathrm{PH}}\!\bigl(\widehat\nu_{1,k}\bigr)
-
\Psi_{\mathrm{PH}}\!\bigl(\widehat\nu_{0,k}\bigr)
}_{E}.
\]
This quantity estimates the binwise causal topological contrast under conditional topological ignorability relative to \(\Psi_{\mathrm{PH}}\), even though the chosen persistent-homology summary is generally noninjective and therefore does not identify the full interventional laws.

The usual stability theory for persistent homology provides the link between
statistical estimation of the underlying filtration and estimation of the
topological signature.  For example, if the filtration functions used to build
the diagrams are close in sup-norm, then the associated persistence diagrams
are close in bottleneck distance
\cite{CohenSteinerEdelsbrunnerHarer2007}. Persistence stability does not imply that Betti numbers at every fixed level
are exactly unchanged. Rather, it implies that birth and death coordinates can be matched within the perturbation scale in bottleneck distance. Consequently, Betti numbers at a
chosen level are stable provided that level is separated from nearby
birth/death values by a sufficient margin. This is why the experiments below report finite-grid summaries together with
robustness diagnostics, rather than claiming exact equality of the entire
continuum of persistence modules.

Two practical remarks matter. First, the choice of filtration is part of the estimand. When \(\widehat\nu_{t,k}\) is represented by a finite point cloud, the resulting persistent homology is a sample-based approximation to the intended population-level filtration. A persistence landscape built from a kernel density level-set filtration targets a different scientific feature from one built from a Vietoris--Rips filtration on the raw point cloud. Second, the vectorization is also part of the estimand. Replacing a persistence landscape by a persistence image or a Betti curve changes the causal target, not merely the numerical implementation. The theory above is therefore intentionally written for a generic transform $\Psi$ so that the causal interpretation stays attached to the actual summary used in the data analysis.

\subsection{Overlap diagnostics and sensitivity}

Causal estimation should be accompanied by the same diagnostic discipline used
for ATE/CATE analyses.  Accordingly, an implementation should report the fitted
propensity-score range, the fraction of observations outside \([0.1,0.9]\),
minimum treated and control counts within score bins, and arm-specific
effective sample sizes.  For weights \(w_i\), the effective sample size in arm
\(t\) is
\[
\ESS_t
=
\frac{\bigl(\sum_{i:T_i=t} w_i\bigr)^2}
{\sum_{i:T_i=t} w_i^2}.
\]
In our implementation, \(w_i\) denotes the stabilized ATE weight based on the
estimated propensity score \(\widehat e(Z_i)\):
\[
w_i
=
\begin{cases}
	\widehat\pi/\widehat e(Z_i), & T_i=1,\\[1mm]
	(1-\widehat\pi)/(1-\widehat e(Z_i)), & T_i=0,
\end{cases}
\qquad
\widehat\pi=\frac1n\sum_{i=1}^n T_i .
\]
The same diagnostics can also be computed for unstabilized IPW weights.  These
overlap diagnostics do not test hidden ignorability, but they indicate whether
the observed-covariate weighting problem is numerically well posed or driven by
extreme weights and poor common support.

To calibrate sensitivity to hidden bias, we use an odds-tilt family.  For
\(\Gamma\ge1\), define perturbed propensity scores by
\[
\logit\bigl(e_{\Gamma}^{\pm}(z)\bigr)
=
\logit\bigl(e(z)\bigr)\pm\log\Gamma .
\]
Equivalently, the treatment odds are multiplied or divided by \(\Gamma\):
\[
\frac{e_{\Gamma}^{+}(z)}{1-e_{\Gamma}^{+}(z)}
=
\Gamma\frac{e(z)}{1-e(z)},
\qquad
\frac{e_{\Gamma}^{-}(z)}{1-e_{\Gamma}^{-}(z)}
=
\Gamma^{-1}\frac{e(z)}{1-e(z)}.
\]
For each value of \(\Gamma\) and each sign, we recompute the stabilized IPW
weights, the coordinate-mean contrast, and the selected finite superlevel
topological contrasts.  The two tilted values, or their worst-case range,
summarize how sensitive the reported contrast is to controlled perturbations of
the treatment odds.

This procedure is not a sharp identified set for arbitrary hidden confounding.
Rather, it is a transparent robustness diagnostic: if the selected
\(\beta_0\), \(\beta_1\), or Euler summaries remain stable across a range of
odds tilts, while the coordinate-mean target remains sensitive, this supports
the interpretation that the reported topological contrast is not merely an
artifact of one fitted propensity model.

\FloatBarrier

\section{Finite Superlevel-Set Experiments}
\label{sec:superlevel-experiments}

We now report three experiments based on finite density-superlevel summaries.  The first two are
hidden-confounding structural benchmarks: one fully synthetic and one
semi-synthetic.  The third is a randomized JOBS II application used to
illustrate the same finite-superlevel summaries on real experimental data.

In all experiments, topology is computed from smoothed empirical density grids on the native two-dimensional outcome space.  For a finite grid of density levels \(\Lambda\), we compute the Betti curves
\[
\bigl(\beta_0(\{\widehat f\ge c\})\bigr)_{c\in\Lambda},
\qquad
\bigl(\beta_1(\{\widehat f\ge c\})\bigr)_{c\in\Lambda},
\]
together with the Euler curve
\[
\bigl(\chi(\{\widehat f\ge c\})\bigr)_{c\in\Lambda}.
\]
A finite ECT vector is also computed from Euler characteristics of directional slices of the same superlevel sets.  However, ECT is treated as a descriptive location-sensitive diagnostic unless the slice-wise chamber-invariance check passes exactly.

\subsection{Hidden-confounding synthetic and semi-synthetic benchmarks}
The goal of these experiments is not to show that topological summaries are
universally superior to mean targets.  Rather, the goal is to illustrate the
identification logic developed above: a noninjective structural target can
remain stable under hidden-confounding reweighting even when a coordinate-mean
target remains biased.

The first benchmark is fully synthetic.  The second uses the 30 observed
covariates from the Wisconsin breast-cancer data, while keeping the treatment
mechanism and two-dimensional outcomes semi-synthetic.  In both
hidden-confounding benchmarks, treatment depends on both an observed score and
a hidden binary variable \(U\).
Thus weak ignorability fails by construction.  
The outcome maps are designed so that hidden \(U\) changes the sampling density within the outcome regions while preserving the selected finite superlevel topology: the control outcome retains one connected high-density region, whereas the treated outcome retains two connected high-density regions, with no one-dimensional hole over the selected level grid.

Figure~\ref{fig:superlevel-clouds} shows the observed factual clouds and the
oracle interventional clouds for the two hidden-confounding benchmarks. The figure
displays the native two-dimensional outcome geometry; the balance diagnostics,
mean-bias calculations, and finite-superlevel topology summaries are reported
numerically in Table~\ref{tab:superlevel-main}.

\begin{figure}[htbp]
	\centering
	
	\begin{minipage}{0.48\textwidth}
		\centering
		\includegraphics[width=\textwidth]{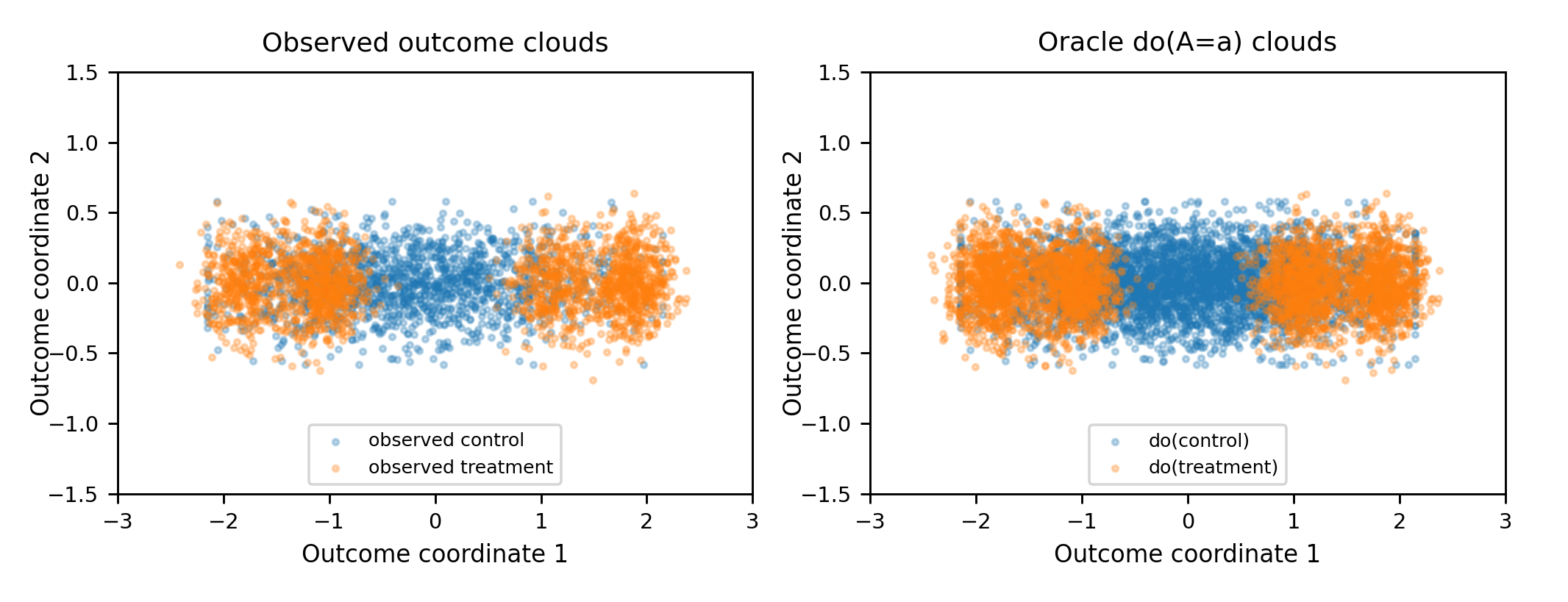}
		\par\smallskip
		\textbf{(a)} Synthetic exact benchmark.
	\end{minipage}
	\hfill
	\begin{minipage}{0.48\textwidth}
		\centering
		\includegraphics[width=\textwidth]{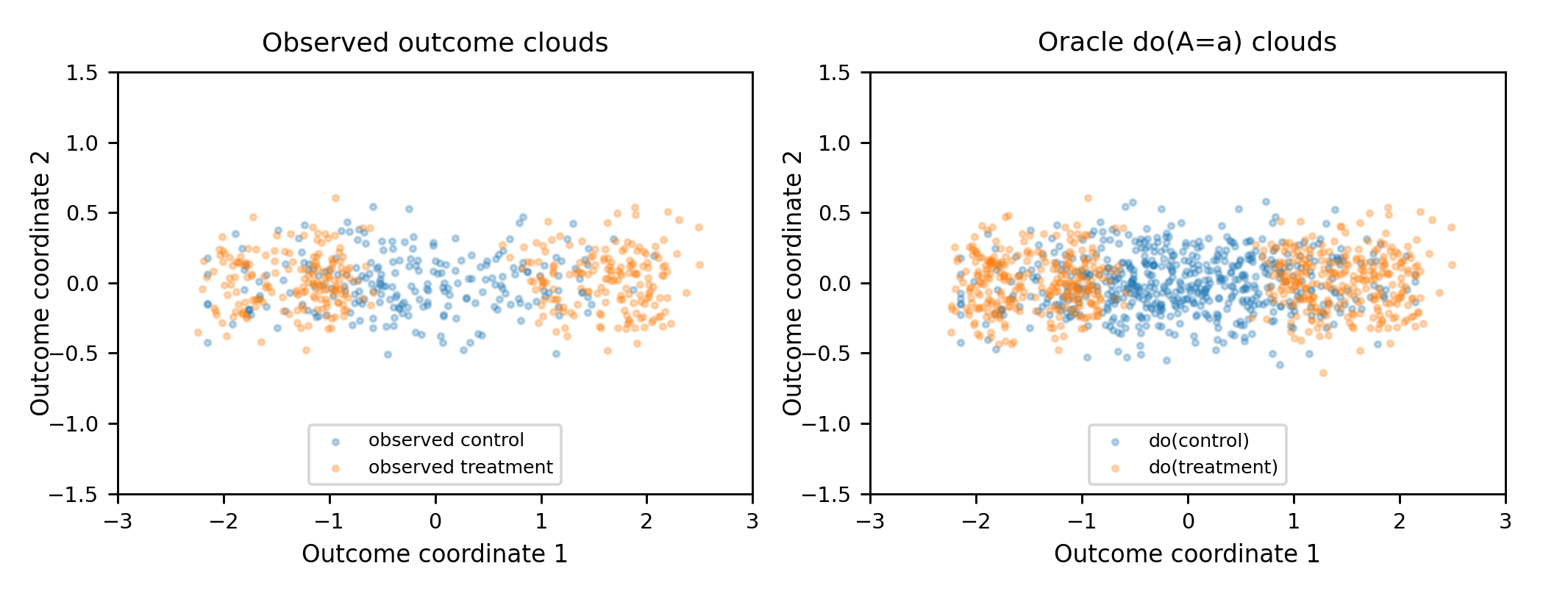}
		\par\smallskip
		\textbf{(b)} Breast-cancer semi-synthetic benchmark.
	\end{minipage}
	
	\caption{
		Observed factual outcome clouds and oracle interventional outcome clouds
		for the two hidden-confounding  benchmarks. The point
		clouds visualize the geometry underlying the finite-superlevel summaries:
		the control distribution is designed to have one connected high-density
		region, while the treated distribution is designed to have two. The
		numerical balance, mean-bias, and finite-superlevel topology summaries
		are reported in Table~\ref{tab:superlevel-main}.
	}
	\label{fig:superlevel-clouds}
\end{figure}

Table~\ref{tab:superlevel-main} summarizes the main numerical results.  Observed-covariate weighting nearly eliminates observed imbalance: the maximum absolute SMD decreases from \(0.728\) to \(0.001\) in the synthetic benchmark and from \(0.430\) to \(0.006\) in the breast-cancer semi-synthetic benchmark.  Nevertheless, the coordinate-mean target remains biased after IPW adjustment, with IPW biases \(0.221\) and \(0.236\), respectively.  This confirms the intended hidden-confounding failure mode: observed balance in \(Z\) does not remove the effect of hidden \(U\).

\begin{table}[htbp]
	\centering
	\small
	\caption{Finite superlevel-set hidden-confounding benchmarks.  The mean target is the ATE on the first native outcome coordinate.  IPW uses only observed covariates.}
	\label{tab:superlevel-main}
	\begin{tabular}{lccccccc}
		\toprule
		Dataset & Raw SMD & Weighted SMD & True ATE$_x$ & Naive bias & IPW bias & \(\Delta\beta_0\) ref. & \(\Delta\chi\) ref. \\
		\midrule
		Synthetic exact
		& 0.728 & 0.001 & -0.029 & 0.257 & 0.221 & 1.000 & 1.000 \\
		\shortstack[l]{Breast-cancer\\semi-synthetic}
		& 0.430 & 0.006 & 0.053 & 0.275 & 0.236 & 1.000 & 1.000 \\
		\bottomrule
	\end{tabular}
\end{table}

The balancing-bin standardized topology confirms the same conclusion.  In both hidden-confounding benchmarks, the observed and IPW finite-superlevel summaries
agree with the oracle effect.  For the reported finite-grid summaries,
\[
\Delta\beta_0=1,
\qquad
\Delta\beta_1=0,
\qquad
\Delta\chi=1.
\]
Thus the selected noninjective topological contrast is stable even though the
full-law equality required by weak ignorability fails and the coordinate-mean
target remains biased.

The finite ECT summaries are more sensitive.  In the synthetic benchmark, only
\(85.6\%\) of global ECT slice coordinates agree exactly with the oracle effect,
and in the breast-cancer semi-synthetic benchmark this fraction is \(72.2\%\).
Therefore, ECT is reported as a descriptive shape/location diagnostic, not as
part of the formal empirical CTI claim in these two benchmarks. 

Finally, we performed several robustness diagnostics for the two
hidden-confounding benchmarks.  These include overlap diagnostics, odds-tilt
perturbations of the fitted propensity score, nonparametric bootstrap
resampling, and finite-signature tuning over grid size, smoothing, level range,
ECT directions and slices, and small-component thresholds.  In both benchmarks,
the reference \(\Delta\beta_0\) and \(\Delta\chi\) effects remained equal to
\(1\) under the odds-tilt stress test up to \(\Gamma=5\).  The finite-signature
tuning also preserved the observed and oracle \(\Delta\beta_0\) reference
effects over the reported grid.  These diagnostics support the stability of
the selected finite Betti and Euler summaries, while ECT remains a more
sensitive descriptive signature.  They should be interpreted as robustness
diagnostics, not as a sharp sensitivity bound for arbitrary hidden confounding.
\subsection{Randomized JOBS II application}

The third experiment uses the JOBS II randomized intervention data.  This
analysis is not a hidden-confounding proof benchmark, because unit-level oracle
potential outcomes are not observed.  Instead, it is a real randomized
application showing how finite superlevel-set summaries can complement
conventional mean contrasts.

The analysis retains \(899\) randomized participants, with \(600\) assigned to
the intervention and \(299\) to control.  The two-dimensional outcome consists of standardized job-search self-efficacy
and standardized negative depression score, so larger values correspond to
stronger job-search self-efficacy and lower depression. The mean contrasts are modest and positive: the randomized
difference is \(0.093\) for job-search self-efficacy and \(0.097\) for negative
depression.  IPW and overlap adjustments give similar values.
Figure~\ref{fig:jobs-ii-superlevel} visualizes the observed joint outcome cloud
and the corresponding finite ECT contrast.

\begin{figure}[htbp]
	\centering
	
	\begin{minipage}{0.48\textwidth}
		\centering
		\includegraphics[width=\textwidth]{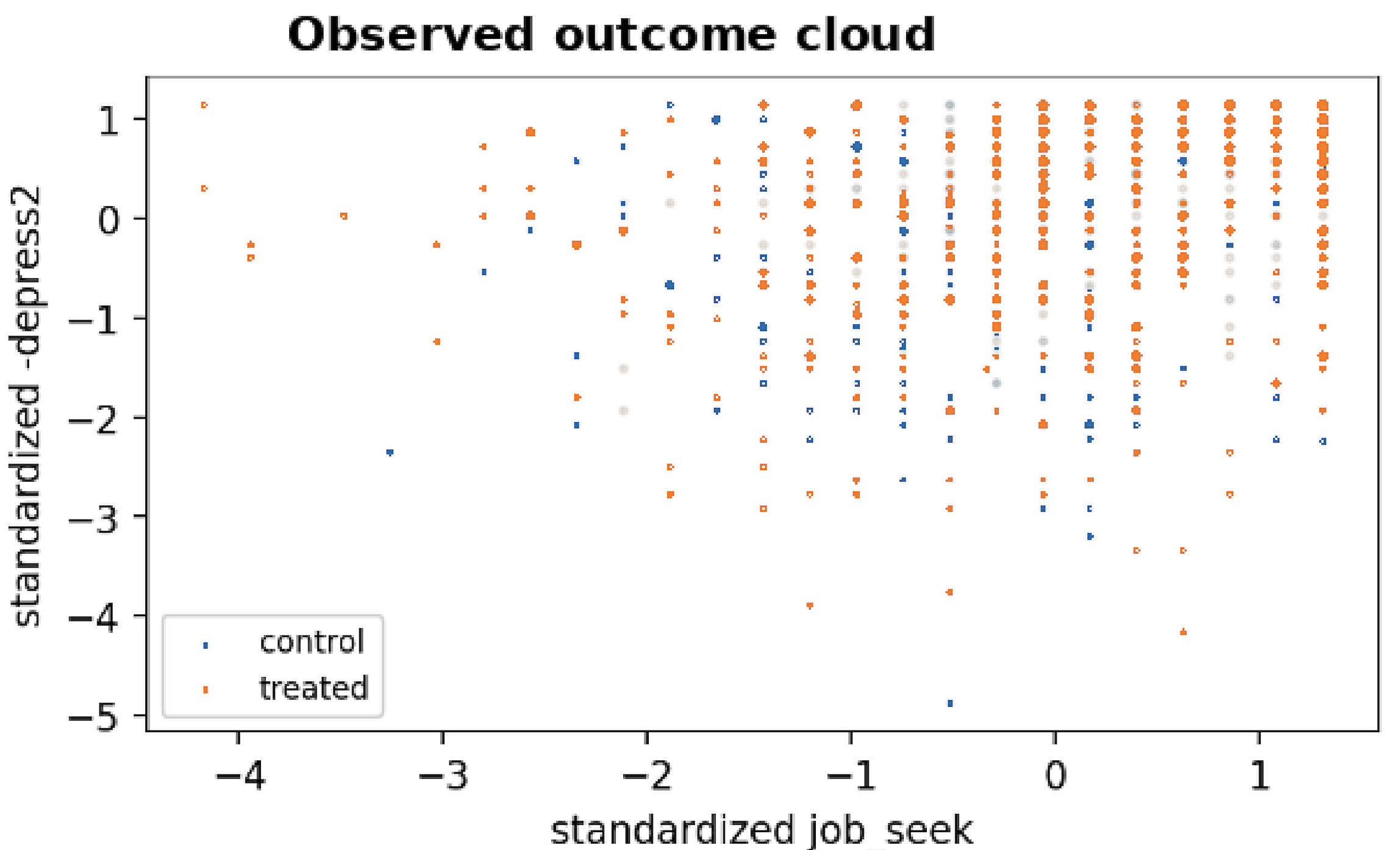}
		\par\smallskip
		\textbf{(a)} Observed outcome cloud.
	\end{minipage}
	\hfill
	\begin{minipage}{0.48\textwidth}
		\centering
		\includegraphics[width=\textwidth]{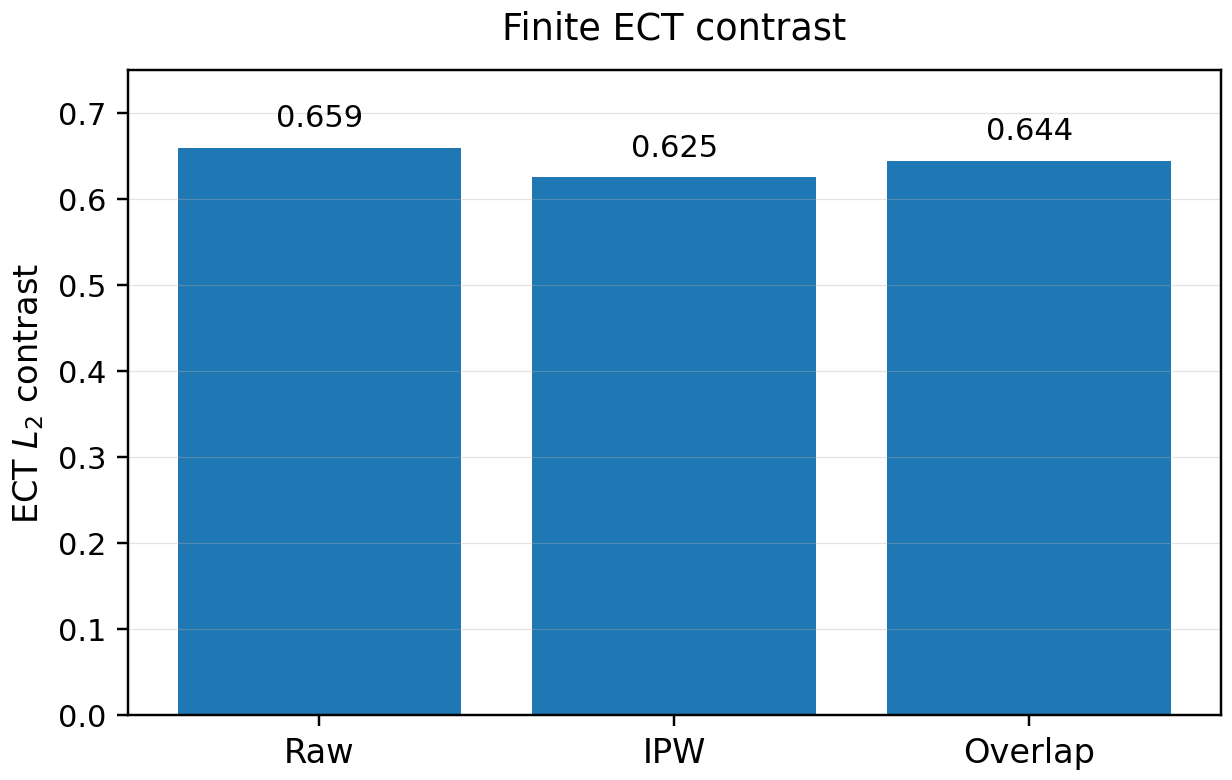}
		\par\smallskip
		\textbf{(b)} Finite ECT contrast.
	\end{minipage}
	
\caption{
	JOBS II randomized application.  Panel (a) shows the observed joint
	outcome cloud, with standardized job-search self-efficacy on the
	horizontal axis and standardized negative depression score on the vertical
	axis.  Panel (b) reports the finite ECT contrast under raw, IPW, and
	overlap analyses.  The ECT signal suggests a distributional
	shape/location difference between intervention and control groups,
	complementing the modest mean contrasts reported in the text.
}
	\label{fig:jobs-ii-superlevel}
\end{figure}

Consistent with Figure~\ref{fig:jobs-ii-superlevel}, the finite superlevel
summaries show additional distributional structure. 
Globally, the reference-level Betti-0 and Euler effects are both \(1.000\),
while the finite ECT contrast remains nonzero under raw, IPW, and overlap
analyses.  However, the score-bin results show heterogeneity: one bin has a
negative reference Betti-0 effect, while the other bins have no reference-level
component change.  Therefore, in JOBS II, the safest interpretation is not a
uniform component-count effect, but rather a randomized distributional
difference in the joint outcome density, especially visible through the
directional ECT signature.

Overall, the synthetic and breast-cancer semi-synthetic benchmarks support the
main theoretical point: weak ignorability can fail and mean targets can remain
biased even after strong observed-covariate balance, while selected finite
density-superlevel summaries remain invariant under hidden-confounding
reweighting.  The JOBS II analysis shows how the same finite-superlevel methodology can be
used as a descriptive tool in a real randomized application.

\FloatBarrier

\section{Discussion}

The main conclusion is sharp.  Topological ignorability is not a single
universal relaxation of weak ignorability.  It is a \emph{target-dependent}
concept.  If the transform \(\Psi\) is injective on the relevant model class,
topological ignorability and weak ignorability coincide.  In that regime,
topology changes the causal \emph{estimand}, not the identifying assumptions.
If \(\Psi\) is noninjective but scientifically well aligned, conditional
topological ignorability can be strictly weaker.  The worked hidden-confounding
example and the finite superlevel-set experiments make this point explicit:
 weak ignorability fails and mean targets can remain biased, yet selected density-superlevel Betti and Euler summaries can remain identifiable
 under the target-dependent topological ignorability condition, because the
 chosen functional forgets density changes that do not alter the selected
 superlevel-set topology.

The paper also clarifies what should be estimated in the noninjective regime.
One should generally not claim identification of the transformed marginal
interventional laws.  The causal quantity directly justified by conditional
topological ignorability is the covariate-standardized effect
\[
\tau_{\Psi}
=
\E\!\left[
\norm{\Psi(P_{1,Z})-\Psi(P_{0,Z})}_E
\right],
\]
which is estimated by balancing on covariates, computing factual arm-specific
topological signatures within strata, score bins, or neighborhoods, and
averaging those contrasts.  This is the natural template for
persistent-homology-type summaries, Betti curves, Euler curves, and other
finite topological signatures that are meaningful but noninjective.

This distinction matters in applications because many structural questions are
coarse by nature.  In a cell-state application, one may care about whether a
perturbation creates a new branch or disconnected high-density regime, not
about the exact density on each branch.  In a social-policy application, one
may care about whether treatment reorganizes the joint response cloud into
structurally different outcome regions, not only about a mean depression score.
In such settings, insisting on identification of the full law can be stronger
than the scientific question requires.

At the same time, topology does not suspend ordinary causal discipline.  One
still needs overlap, careful estimation of conditional laws, finite-sample
diagnostics, and robustness checks for hidden bias and numerical tuning.  The
experiments illustrate this distinction: observed-covariate weighting can
nearly eliminate standardized mean differences while hidden confounding still
biases the coordinate-mean target; nevertheless, selected finite Betti and
Euler summaries may remain stable.  More sensitive signatures such as finite
ECT can reveal additional shape or location differences, but they should be
included in a formal topological-ignorability claim only when the corresponding
slice-wise invariance checks are satisfied.

A limitation is that conditional topological ignorability is generally not
testable from a single observational distribution, just as weak ignorability is
not testable without additional design or structural assumptions.  The SEM
criteria and experiments should therefore be interpreted as mechanisms and
diagnostics, not as automatic tests of the assumption.  In applications, the
choice of \(\Psi\) must be justified scientifically, and robustness should be
assessed over density estimation, level choices, smoothing, weighting, and
hidden-bias perturbations.

The proposed framework therefore does not replace weak ignorability.  It
refines the question: for a scientifically chosen noninjective structural
summary, one may identify a target-level causal contrast without identifying
the full interventional distribution.  This is the sense in which topological
ignorability is useful.

\appendix
\section{Nerve criterion for loops and higher holes}
\label{app:nerve-criterion}

For loops and higher-dimensional holes, connected-component information is not enough. 
One must keep track of higher-order overlaps between pieces.  The nerve of a good cover provides a combinatorial model for this information.

\begin{definition}[Good cover]
	A finite open cover
	\[
	\mathcal U=\{U_i\}_{i\in I}
	\]
	of a topological space \(X\) is called a \emph{good cover} if every nonempty finite intersection
	\[
	U_{i_1}\cap\cdots\cap U_{i_r}
	\]
	is contractible.
\end{definition}

\begin{definition}[Nerve of a cover]
	Given a finite cover
	\[
	\mathcal U=\{U_i\}_{i\in I},
	\]
	its nerve \(N(\mathcal U)\) is the simplicial complex whose vertices correspond to the sets \(U_i\), and for which
	\[
	[i_0,\ldots,i_r]
	\]
	forms an \(r\)-simplex whenever
	\[
	U_{i_0}\cap\cdots\cap U_{i_r}\neq\emptyset.
	\]
\end{definition}

\begin{theorem}[Nerve theorem]
	If a paracompact space \(X\) admits a finite good cover \(\mathcal U\), then \(X\) is homotopy equivalent to the nerve \(N(\mathcal U)\); see, e.g., Björner~\cite{Bjorner}.
\end{theorem}

\begin{theorem}[Nerve criterion for loops and higher holes]
	\label{thm:sem-nerve-criterion}
	Fix \((t,z)\), a finite level set \(\Lambda=\{c_1>\cdots>c_m\}\), and homological dimensions \(\mathcal D\subseteq\{1,2,\ldots\}\).  For each \(c\in\Lambda\), suppose that
	\[
	A_c=A^{t,z}_{c}(\Pi_z),
	\qquad
	A_c^{\mathrm{obs}}=A^{t,z}_{c}(\Pi^{\mathrm{obs}}_{t,z})
	\]
	admit finite good covers with nerves \(N_c\) and \(N_c^{\mathrm{obs}}\).
	Assume moreover that the chosen good covers are compatible with the
	inclusions
	\[
	A_{c_i}\subseteq A_{c_j},
	\qquad
	A^{\mathrm{obs}}_{c_i}\subseteq A^{\mathrm{obs}}_{c_j}
	\qquad
	\text{whenever } c_i\ge c_j,
	\]
	so that these inclusions induce simplicial maps between the corresponding
	nerves.
	
	Suppose that, for each \(k\in\mathcal D\), the homology groups
	\[
	H_k(N_c^{\mathrm{obs}};\mathbb F)
	\quad\text{and}\quad
	H_k(N_c;\mathbb F)
	\]
	are isomorphic at every \(c\in\Lambda\), and that these isomorphisms commute
	with the inclusion-induced maps across levels.  Then the factual and interventional persistent \(H_k\) modules are isomorphic for every
	\(k\in\mathcal D\).  In particular, the Betti curves and finite barcodes in
	those dimensions agree.
\end{theorem}

\begin{proof}
	By the nerve theorem, each level set is homotopy equivalent to the nerve of its finite good cover. Therefore, for each \(c\in\Lambda\) and \(k\in\mathcal D\),
	\[
	H_k(A_c;\mathbb F)\cong H_k(N_c;\mathbb F),
	\qquad
	H_k(A_c^{\mathrm{obs}};\mathbb F)\cong H_k(N_c^{\mathrm{obs}};\mathbb F).
	\]
	The compatibility of the covers across levels ensures that the inclusion maps of the filtrations correspond to the induced simplicial maps between nerves. The assumed commuting isomorphisms of nerve homology therefore give an isomorphism of the corresponding persistence modules. Hence the persistent \(H_k\) modules agree, and so the associated Betti curves and finite barcodes agree.
\end{proof}

\begin{corollary}[Loop-based conditional topological ignorability]
	\label{cor:sem-h1-ignorability}
	Taking \(\mathcal D=\{1\}\), if Theorem~\ref{thm:sem-nerve-criterion} holds for \(t=0,1\) and for \(P_Z\)-almost every \(z\), then conditional topological ignorability holds for every \(\Psi\) depending only on the selected persistent \(H_1\) summary. Geometrically, reweighting may change density along the loop, but it must not break the loop, fill it, or create an additional loop at the levels used by \(\Psi\).
\end{corollary}

\begin{example}[A loop that is preserved or destroyed]
	\label{ex:h1-loop}
	Place four contractible high-density pieces around the corners of a square in \(\mathbb R^2\). Suppose adjacent pieces overlap, opposite pieces do not, and no triple intersection occurs. The nerve is a four-cycle, so the union has one \(H_1\) loop. Reweighting that keeps all four pieces active and preserves the same overlap pattern leaves the nerve unchanged and preserves the loop. If treatment selection suppresses one piece below the chosen density level, the nerve becomes a path and the loop disappears. Thus the same SEM may violate weak ignorability in both cases, but \(H_1\)-based topological ignorability holds only in the first case; see Figure~\ref{nerve-loop}.
\end{example}

\begin{figure}
	\centering
	\includegraphics[scale=0.2]{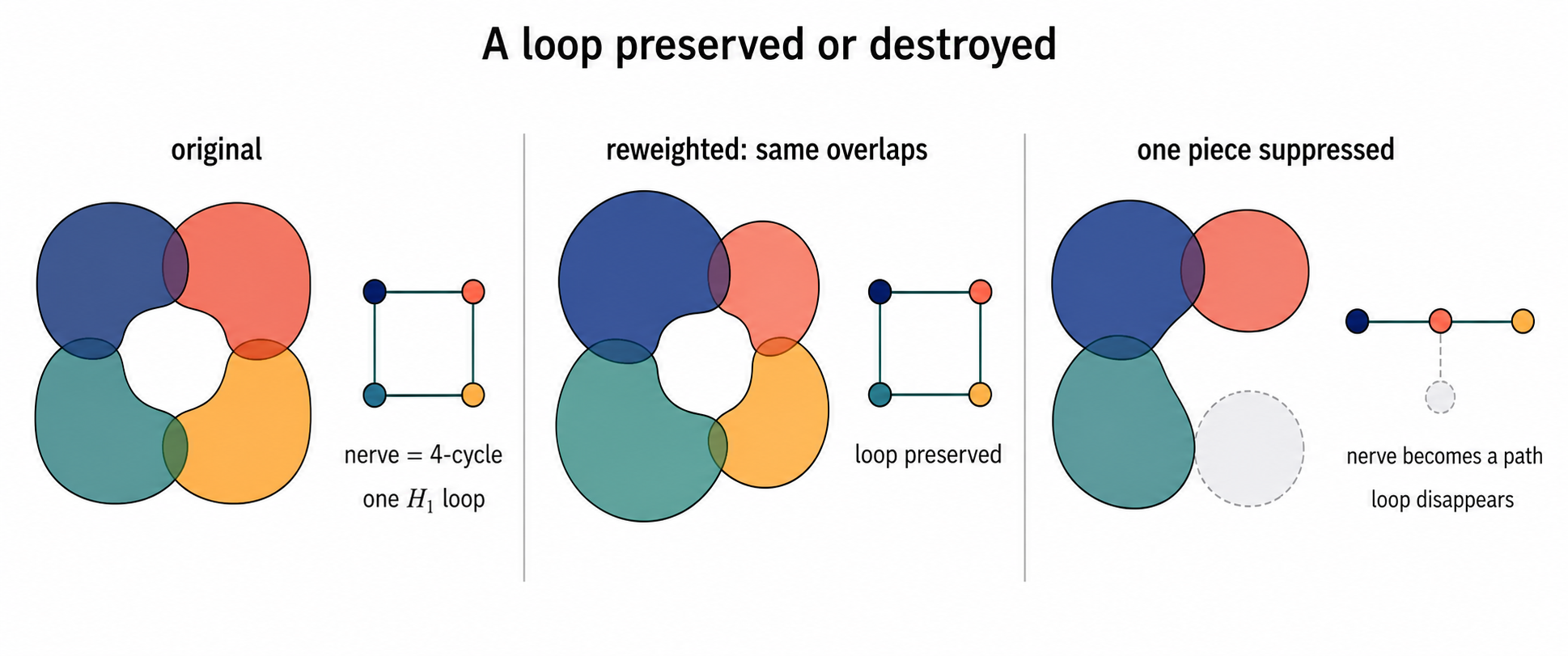}
	\caption{
		A four-cycle nerve preserves one \(H_1\) loop under reweightings that keep the same active pieces and overlaps; suppressing one piece turns the nerve into a path and destroys the loop.
	}
	\label{nerve-loop}
\end{figure}

\section{Full construction of the worked SEM example}
\label{app:worked-sem}

This appendix gives the full structural construction behind the worked example in the main text.  The purpose is to show that the densities used there arise from an explicit hidden-confounding SEM.

Let \(U\in\{0,1\}\) be an unobserved latent variable with
\[
\Prob(U=0)=\Prob(U=1)=\frac12,
\]
and let \(V,E_0,E_1\overset{ind}{\sim}\mathrm{Unif}(0,1)\), independent of \(U\). Treatment is assigned by
\[
T=\mathbf 1\{V\le p(U)\},
\qquad
p(0)=\frac15,
\qquad
p(1)=\frac45.
\]
Thus \(U\) affects treatment assignment.

Potential outcomes are generated by monotone piecewise-affine maps:
\[
Y(0)=Q_{0,U}(E_0),
\qquad
Y(1)=Q_{1,U}(E_1),
\qquad
Y=(1-T)Y(0)+TY(1).
\]
For \(a,b\in\{0,1\}\), write
\[
\mu_{a,b}:=\Law(Y(a)\mid U=b),
\qquad
f_{a,b}:=\frac{d\mu_{a,b}}{dy}.
\]
Also write
\[
\mu_a:=\Law(Y(a)),
\qquad
f_a:=\frac{d\mu_a}{dy},
\qquad
\nu_a:=\Law(Y\mid T=a),
\qquad
h_a:=\frac{d\nu_a}{dy}.
\]
All displayed density equalities are understood up to equality almost everywhere.

The structural maps are
\[
Q_{0,0}(e)=3e,
\qquad 0\le e\le 1,
\]
\[
Q_{0,1}(e)=
\begin{cases}
	5e, & 0\le e\le \frac15,\\[2mm]
	1+\frac{10}{3}\left(e-\frac15\right), & \frac15< e\le \frac12,\\[2mm]
	2+2\left(e-\frac12\right), & \frac12< e\le 1,
\end{cases}
\]
\[
Q_{1,0}(e)=
\begin{cases}
	2e, & 0\le e\le \frac12,\\[2mm]
	2+2\left(e-\frac12\right), & \frac12< e\le 1,
\end{cases}
\]
and
\[
Q_{1,1}(e)=
\begin{cases}
	5e, & 0\le e\le \frac15,\\[2mm]
	2+\frac54\left(e-\frac15\right), & \frac15< e\le 1.
\end{cases}
\]

Since \(Q_{a,0}\neq Q_{a,1}\) for both \(a=0,1\), the law of \(Y(a)\) depends on \(U\) in both arms. Therefore \(U\) affects both treatment assignment and potential outcomes, so it is a genuine hidden confounder.

We use the following elementary change-of-variables fact. If \(E\sim \mathrm{Unif}(0,1)\) and \(Q\) is affine with slope \(s>0\) on an interval whose image is \(J\), then that branch contributes
\[
\frac1s\mathbf 1_J(y)
\]
to the density of \(Q(E)\). Applying this branchwise gives the following latent-state-specific densities.

\begin{table}[htbp]
	\centering
	\small
	\renewcommand{\arraystretch}{1.15}
	\caption{Latent-state-specific densities generated by the structural maps in the worked SEM.}
	\label{tab:app-example-densities}
	\begin{tabular}{cccc}
		\toprule
		Arm \(a\) & Latent state \(b\) & Density \(f_{a,b}(y)\) & \(\E[Y(a)\mid U=b]\) \\
		\midrule
		0 & 0
		&
		\(\frac13\mathbf 1_{[0,3]}(y)\)
		&
		\(\frac32\)
		\\[1mm]
		0 & 1
		&
		\(\frac15\mathbf 1_{[0,1]}(y)
		+\frac{3}{10}\mathbf 1_{[1,2]}(y)
		+\frac12\mathbf 1_{[2,3]}(y)\)
		&
		\(\frac95\)
		\\[1mm]
		1 & 0
		&
		\(\frac12\mathbf 1_{[0,1]}(y)
		+\frac12\mathbf 1_{[2,3]}(y)\)
		&
		\(\frac32\)
		\\[1mm]
		1 & 1
		&
		\(\frac15\mathbf 1_{[0,1]}(y)
		+\frac45\mathbf 1_{[2,3]}(y)\)
		&
		\(\frac{21}{10}\)
		\\
		\bottomrule
	\end{tabular}
\end{table}

Since
\[
\Prob(T=1)
=
\frac12\cdot\frac15+\frac12\cdot\frac45
=
\frac12,
\]
we also have \(\Prob(T=0)=1/2\). The interventional densities are therefore
\[
f_0(y)
=
\frac12 f_{0,0}(y)+\frac12 f_{0,1}(y)
=
\frac{4}{15}\mathbf 1_{[0,1]}(y)
+
\frac{19}{60}\mathbf 1_{[1,2]}(y)
+
\frac{5}{12}\mathbf 1_{[2,3]}(y),
\]
and
\[
f_1(y)
=
\frac12 f_{1,0}(y)+\frac12 f_{1,1}(y)
=
\frac{7}{20}\mathbf 1_{[0,1]}(y)
+
\frac{13}{20}\mathbf 1_{[2,3]}(y).
\]
By consistency, on the event \(\{T=a\}\) we have \(Y=Y(a)\). Therefore
\(\nu_a=\Law(Y\mid T=a)=\Law(Y(a)\mid T=a)\).
Bayes' rule gives the factual latent weights
\[
\Prob(U=0\mid T=1)
=
\frac{\Prob(T=1\mid U=0)\Prob(U=0)}{\Prob(T=1)}
=
\frac{(1/5)(1/2)}{1/2}
=
\frac15,
\]
\[
\Prob(U=1\mid T=1)=\frac45,
\]
and similarly
\[
\Prob(U=0\mid T=0)=\frac45,
\qquad
\Prob(U=1\mid T=0)=\frac15.
\]
Hence the factual densities are
\[
h_0(y)
=
\frac45 f_{0,0}(y)+\frac15 f_{0,1}(y)
=
\frac{23}{75}\mathbf 1_{[0,1]}(y)
+
\frac{49}{150}\mathbf 1_{[1,2]}(y)
+
\frac{11}{30}\mathbf 1_{[2,3]}(y),
\]
and
\[
h_1(y)
=
\frac15 f_{1,0}(y)+\frac45 f_{1,1}(y)
=
\frac{13}{50}\mathbf 1_{[0,1]}(y)
+
\frac{37}{50}\mathbf 1_{[2,3]}(y).
\]

These are precisely the densities used in the worked example in the main text. They show explicitly how conditioning on factual treatment changes the latent weights and hence changes the density values, while leaving the selected superlevel-set topology unchanged at the density level used in the example.

\section{A second SEM with the same observational law and a different ATE}
\label{app:same-observational-different-ate}

This appendix gives a second structural model \(M'\) with the same observational law of \((T,Y)\) as the worked SEM \(M\), but with a different average treatment effect.  This shows that the observational law does not identify the ATE on this confounded model class, even though the selected topological contrast in the worked example can remain identified.

The purpose of \(M'\) is not to change the main topological example, but to show that the same observed data can be compatible with a different mean causal effect. We therefore use \(M'\) only to demonstrate nonidentification of the ATE, not as a second example of preserved superlevel topology.

Model \(M'\) uses the same latent variable \(U\), the same exogenous variables \(V,E_0,E_1\), and the same treatment rule
\[
T=\mathbf 1\{V\le p(U)\},
\qquad
p(0)=\frac15,
\qquad
p(1)=\frac45.
\]
Only the potential-outcome maps are changed. Define
\[
Q'_{1,0}(e)=2+e,
\qquad 0\le e\le 1,
\]
\[
Q'_{1,1}(e)=
\begin{cases}
	\frac{40}{13}e, & 0\le e\le \frac{13}{40},\\[2mm]
	2+\frac{40}{27}\left(e-\frac{13}{40}\right), & \frac{13}{40}<e\le 1,
\end{cases}
\]
\[
Q'_{0,0}(e)=
\begin{cases}
	\frac{50}{19}e, & 0\le e\le \frac{19}{50},\\[2mm]
	1+5\left(e-\frac{19}{50}\right), & \frac{19}{50}<e\le \frac{29}{50},\\[2mm]
	2+\frac{50}{21}\left(e-\frac{29}{50}\right), & \frac{29}{50}<e\le 1,
\end{cases}
\]
and
\[
Q'_{0,1}(e)=
\begin{cases}
	75e, & 0\le e\le \frac{1}{75},\\[2mm]
	1+\frac65\left(e-\frac{1}{75}\right), & \frac{1}{75}<e\le \frac{127}{150},\\[2mm]
	2+\frac{150}{23}\left(e-\frac{127}{150}\right), & \frac{127}{150}<e\le 1.
\end{cases}
\]

The induced latent-state-specific densities are
\[
g_{1,0}(y)=\mathbf 1_{[2,3]}(y),
\]
\[
g_{1,1}(y)
=
\frac{13}{40}\mathbf 1_{[0,1]}(y)
+
\frac{27}{40}\mathbf 1_{[2,3]}(y),
\]
\[
g_{0,0}(y)
=
\frac{19}{50}\mathbf 1_{[0,1]}(y)
+
\frac15\mathbf 1_{[1,2]}(y)
+
\frac{21}{50}\mathbf 1_{[2,3]}(y),
\]
and
\[
g_{0,1}(y)
=
\frac{1}{75}\mathbf 1_{[0,1]}(y)
+
\frac56\mathbf 1_{[1,2]}(y)
+
\frac{23}{150}\mathbf 1_{[2,3]}(y).
\]

Because the treatment mechanism is unchanged, the posterior weights
\[
\Prob(U=b\mid T=a)
\]
are the same as under \(M\). Thus
\[
\Prob(U=0\mid T=1)=\frac15,
\qquad
\Prob(U=1\mid T=1)=\frac45,
\]
and
\[
\Prob(U=0\mid T=0)=\frac45,
\qquad
\Prob(U=1\mid T=0)=\frac15.
\]

\begin{proposition}[Same observational law but different ATE]
	\label{prop:second-sem}
	Under model \(M'\),
	\[
	\Law_{M'}(T,Y)=\Law_M(T,Y),
	\qquad
	\ATE(M')=\frac{117}{200}\neq \frac{3}{20}=\ATE(M).
	\]
	Consequently, the observational law of \((T,Y)\) does not identify the ATE on the model class containing \(M\) and \(M'\).
\end{proposition}

\begin{proof}
	Using the same posterior weights as under \(M\), we obtain
	\[
	h_1'(y)
	=
	\frac15 g_{1,0}(y)+\frac45 g_{1,1}(y)
	=
	\frac{13}{50}\mathbf 1_{[0,1]}(y)
	+
	\frac{37}{50}\mathbf 1_{[2,3]}(y)
	=
	h_1(y),
	\]
	and
	\[
	h_0'(y)
	=
	\frac45 g_{0,0}(y)+\frac15 g_{0,1}(y)
	=
	\frac{23}{75}\mathbf 1_{[0,1]}(y)
	+
	\frac{49}{150}\mathbf 1_{[1,2]}(y)
	+
	\frac{11}{30}\mathbf 1_{[2,3]}(y)
	=
	h_0(y).
	\]
	Since \(\Prob(T=1)=1/2\) in both models, the joint observational laws of \((T,Y)\) coincide.
	
	It remains to compute the ATE under \(M'\). The means of the latent-state-specific treated laws are
	\[
	\E_{M'}[Y(1)\mid U=0]=\frac52,
	\qquad
	\E_{M'}[Y(1)\mid U=1]=\frac{37}{20}.
	\]
	Therefore
	\[
	\E_{M'}[Y(1)]
	=
	\frac12\cdot\frac52+\frac12\cdot\frac{37}{20}
	=
	\frac{87}{40}.
	\]
	Similarly, the means of the latent-state-specific control laws are
	\[
	\E_{M'}[Y(0)\mid U=0]=\frac{77}{50},
	\qquad
	\E_{M'}[Y(0)\mid U=1]=\frac{41}{25}.
	\]
	Thus
	\[
	\E_{M'}[Y(0)]
	=
	\frac12\cdot\frac{77}{50}
	+
	\frac12\cdot\frac{41}{25}
	=
	\frac{159}{100}.
	\]
	Hence
	\[
	\ATE(M')
	=
	\frac{87}{40}-\frac{159}{100}
	=
	\frac{117}{200}.
	\]
	
	For the original model \(M\), the worked example gives
	\[
	\E_M[Y(1)]
	=
	\frac12\cdot\frac32
	+
	\frac12\cdot\frac{21}{10}
	=
	\frac95,
	\]
	and
	\[
	\E_M[Y(0)]
	=
	\frac12\cdot\frac32
	+
	\frac12\cdot\frac95
	=
	\frac{33}{20}.
	\]
	Therefore
	\[
	\ATE(M)=\frac95-\frac{33}{20}=\frac{3}{20}.
	\]
	Thus \(M\) and \(M'\) have the same observational law but different ATEs.
\end{proof}

This construction shows that the mean effect is not identified from the observational law on this confounded model class. The contrast with the main worked example is that the selected superlevel-set topological contrast in model \(M\) is stable under factual reweighting, even though mean-based causal targets are not identified on the larger confounded model class containing \(M\) and \(M'\).
\end{document}